\documentclass[paper=letter, fontsize=10pt,leqno]{scrartcl}
\usepackage[body={5.35in,7.5in}]{geometry} 

\usepackage[english]{babel}				
\usepackage{amsthm}
\usepackage{hyperref}

 \usepackage{verbatim}
\usepackage[pdftex]{graphicx}				
\usepackage{url}
\usepackage{amssymb}

\usepackage{bbm}
\usepackage{stmaryrd}

\usepackage{dutchcal}
\usepackage{graphics}
\usepackage[all]{xy}

\usepackage[pdftex]{color}
\usepackage{amscd}
\usepackage[format=plain]{caption}
\usepackage{wrapfig}
 \usepackage{float}
 \usepackage{tensor} 
 
 \usepackage{bigints}

\usepackage{sectsty}					
\allsectionsfont{\centering \normalfont\scshape}	

\usepackage{mathtools}
\DeclareCaptionStyle{ruled}{labelfont=bf,labelsep=space}
\usepackage[]{forest}

\newtheorem{theorem}{Theorem}

\newtheorem{proposition}[theorem]{Proposition}
\newtheorem{lemma}[theorem]{Lemma}
\newtheorem{corollary}[theorem]{Corollary}

\newtheorem*{proposition*}{Proposition}

\newtheorem*{assumption*}{Assumption}

\title{
		\vspace{-1in} 	
		\usefont{OT1}{bch}{b}{n}
		\normalfont \normalsize \textsc{} \\ [25pt]
		\huge Chaotic lensed billiards
}

\date{}

\author{\normalfont \large 
\ Timothy Chumley\footnote{\scriptsize Department of Mathematics and Statistics, Mount Hoyoke College, 50 College St., South Hadley, MA 01075},
\ Maeve Covey\footnote{\scriptsize Department of Mathematics and Statistics, Washington University, Campus Box 1146, St. Louis, MO 63130},
\ Christopher Cox\footnotemark[1], 
\  Renato Feres\footnotemark[2]
}

 \date{\today}

\begin{document}

\maketitle

\begin{abstract}
\begin{center}
 Abstract \end{center}
{\small  
Lensed billiards are an extension of the notion of    billiard dynamical systems  obtained by adding a potential function of the form $C\mathbbm{1}_{\mathcal{A}}$, where
$C$ is a real-valued constant and $\mathbbm{1}_{\mathcal{A}}$ is the indicator function of an open subset $\mathcal{A}$  of the billiard table whose boundaries (of $\mathcal{A}$ and the table) are piecewise smooth. 
Trajectories are polygonal lines that undergo either  reflection or refraction at the boundary of $\mathcal{A}$ depending on the angle of incidence. Our main focus is to explore   how the dynamical properties  of these models depend on the potential parameter $C$ using a number of   families of examples.
In particular, we explore numerically the Lyapunov exponents for these parametric families and highlight  the more  salient common properties  that distinguish them from
standard billiard systems.   We further justify some of these properties by characterizing lensed billiards  in terms of switching dynamics  between two  open (standard) billiard subsystems  and
obtaining  mean values associated to orbit sojourn   in each subsystem. 
}
\end{abstract}

\section{Introduction}
Mathematical billiards, particularly in dimension $2$,  are    dynamical systems of a geometric nature that are relatively easy to define and, at the same time,  exhibit  a wide range of dynamic properties, making them useful models for more complicated systems. For this reason they have for many decades  figured prominently in the development of the modern theory of dynamical systems and ergodic theory.   Research on chaotic billiards in dimension $2$, in particular, has by now attained a high degree of technical sophistication. This, in our opinion, justifies an effort to look for generalizations or extensions of the concept of  billiard systems that can point to new directions of research without compromising too much on  the  qualities that make them  attractive  model systems.

\emph{Lensed billiards} may be defined as standard billiard systems to which are added piecewise constant mechanical potential functions of the form $V=C\mathbbm{1}_\mathcal{A}$ where $\mathbbm{1}_\mathcal{A}$ is the indicator function of a subset $\mathcal{A}$ of the billiard domain. We may think of $\mathcal{A}$ as a scatterer that reflects  billiard trajectories that collide with it at sufficiently shallow angles relative to a   {\em critical angle} to be defined shortly, but allows others to pass through and undergo a  {\em refraction}, similar to a light ray crossing the interface surface separating two optical media with different refractive indices.

  The central new  factor to account for   is the  effect of the potential parameter $C$\----the value of the potential function on $\mathcal{A}$\----on the dynamical  behavior of the system. After laying out  basic definitions and the most general properties of  lensed billiards,  the paper focuses on the numerical determination of Lyapunov exponents for a few  parametric families of examples, identifies a number of properties that appear to be common for  these systems,
  and proposes a framework of analysis for explaining these common features.

\begin{figure}[htbp]
\begin{center}
\includegraphics[width=3in]{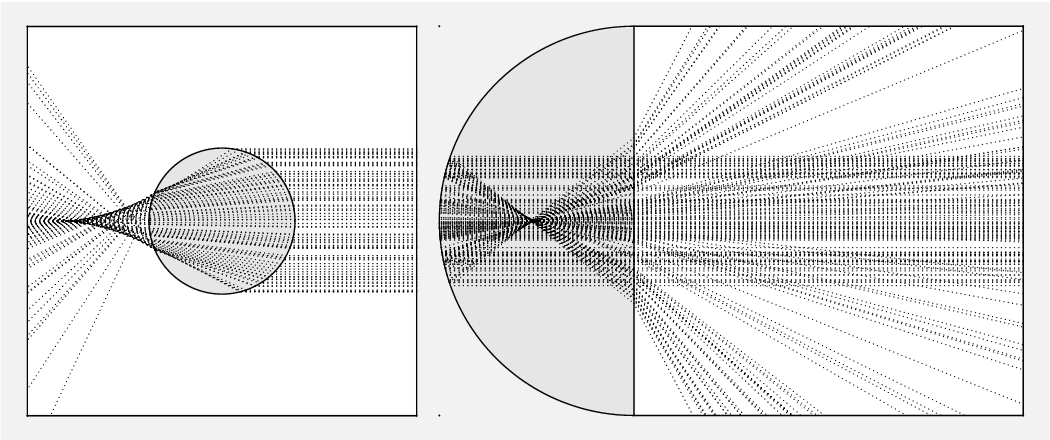}\ \ 
\caption{{\small On the left, a lensed Sinai-type billiard showing focusing and defocusing of trajectories.  On the right, a lensed Bunimovich-like billiard showing initially focusing followed by dispersing trajectories. In both cases, the potential is $-1$ on the shaded region and $0$ outside.  A  parallel  beam of trajectories emanates from the vertical right wall, stopping at the first collision with the boundary of the unshaded region of the billiard domain.}}
\label{SinaiBuni}
\end{center}
\end{figure} 

Before embarking on  a systematic study, let us get a first impression of what these systems look like.
Figure \ref{SinaiBuni} shows two lensed variants of the classical Sinai semi-dispersing and Bunimovich stadium billiards. The shaded regions, let us  denote them by $\mathcal{A}$, indicate where the potential function $V$ is non-zero. In both cases we have set the value of $C$   to $-1$. The negative potential causes trajectories going into $\mathcal{A}$ to deflect as if entering an optical medium with higher refractive index. (As a mechanical system, the speed of the point particle naturally  increases due to conservation of energy\----the opposite of what would happen
to a light ray.) Therefore the Sinai scatterer behaves as a focusing lens. The figure on the left shows the typical focusing-defocusing behavior produced by such a lens. On the other hand, 
trajectories leaving $\mathcal{A}$ on the Bunimovich-type system on the right-hand side exhibit dispersing behavior after standard focusing inside $\mathcal{A}$ by the circular part of the boundary.
This shows how the choice of values of $V$ can alter the standard  mechanisms leading to  hyperbolicity in billiard systems. If the value of the potential in $\mathcal{A}$ is positive then
these behaviors are reversed. It is easy to verify that 
if $C<0$, the system on the left is not ergodic due to trapped trajectories inside the circular lens, and
if $C>0$,   the system on the right is not ergodic due to the existence of a positive measure set of initial conditions initiating in the complement of $\mathcal{A}$ that cannot enter $\mathcal{A}$.  (We consider initial velocities with sufficient energy to allow the billiard particle to cross  the potential barrier, but when the angle of collision
with the vertical line of discontinuity of $V$ is sufficiently small, the particle is nevertheless forced to reflect.)

It is  important to keep in mind the distinction between  the billiard flow and  the billiard map dynamics, due to the fact that particle speed is no longer constant. Figure \ref{Buni_traj} shows the difference in somewhat extreme fashion. As the potential constant $C$ on the semi-disc $\mathcal{A}$   becomes negative with large $|C|$, the number of steps of the billiard trajectories inside $\mathcal{A}$ increases, roughly proportionally to $\sqrt{|C|}$ (a more precise statement will be given later), suggesting that trajectories of the billiard map become increasingly trapped in that  deep potential well. However, 
the speed of the particle in that region is also proportional to $\sqrt{|C|}$, so the amount of time the billiard flow remains in $\mathcal{A}$ is expected to remain bounded.

\begin{figure}[htbp]
\begin{center}
\includegraphics[width=4.5in]{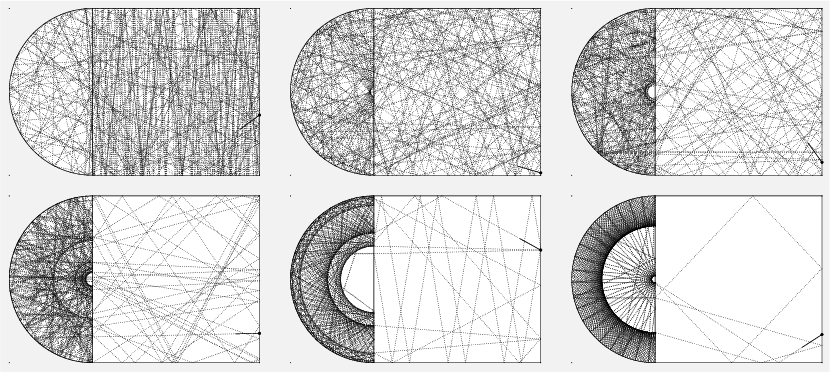}\ \ 
\caption{{\small Sample trajectories (of equal number of steps) for the Lensed Bunimovich billiard. The values $C$ of the potential on the semi-disc, from top to bottom, left to right:  $0$, $-1$, $-10$, $-100$, $-1000$, $-10000$.  Numerical approximation of the positive mean Lyapunov exponent (for the billiard flow) gave the values: $1.00\pm 0.02$ for $C=0$ (top-left) and
$4.46\pm 0.14$ for $C=-10000$ (bottom-right).  See also Figure \ref{Exper2_graph1}.}}
\label{Buni_traj}
\end{center}
\end{figure}

Naturally, the systems we are calling lensed billiards can be viewed as a special case of {\em soft billiards}, or more broadly Hamiltonian systems under the influence of a potential, for which there is a rich existing literature that we now discuss.
It should be emphasized that the present work takes on related systems from a number of new perspectives. A central focus of soft billiards has been for the case of smooth potentials and for the case of a Sinai billiard table consisting of the torus with circular scatterers. Here, the discontinuous potentials of lensed billiards result in impulse-like forces at discrete times and our examples venture beyond the case of the soft Sinai billiard. There are a number of early examples in the literature whose focus is on characterizing ergodicity under appropriate conditions on a smooth potential in the soft Sinai billiard. In \cite{MR0433510}, the author shows that for certain bell-shaped potentials under a general smoothness condition ($C^3$ smoothness at the boundary of the scatterer, among other regularity conditions on the geometry of the table and the total energy of the system), the soft Sinai billiard has the $K$-property and is thus ergodic. Later, in \cite{MR0885572}, the author shows for a class of smooth Coulombic potentials, again for the Sinai billiard on the torus with disk-shaped scatterers, that the flow of the system can be realized as a complete geodesic flow on a suitably defined compact Riemannian manifold. Under appropriate conditions on the potential, the metric has negative curvature and so the flow is Anosov and hence ergodic. In \cite{MR1087385}, these results are generalized to a broader class of smooth potentials. Moreover, under certain conditions on the potential, positivity of Lyapunov exponents is shown using Wojtkowski's method of invariant cone fields. A number of further studies have continued to extend the literature in the case of smooth potentials for the Sinai billiard: exponential decay of correlations and the central limit theorem \cite{MR2020220}, non-ergodicity and stability of periodic orbits \cite{MR1722985, MR1997267}, and related questions for tables in higher dimensions \cite{MR2191384, MR2304468}. It should be noted that the literature for the case of constant potentials seems to be more limited and confined to the case where the table is the Sinai billiard. In \cite{Bal} ergodicity is studied numerically and the parameter space, indexed by the scatterer radius and potential magnitude, is delimited into regions of non-ergodicity. Lyapunov exponents \cite{MR1012710} and ergodicity \cite{MR1162355} have been studied analytically for this example as well. 
Lensed billiards  are also  closely related to so-called {\em composite billiards}, also known as {\em ray-splitting} or {\em branching billiards} \cite{bynk,jssc}, but are  simpler and less of a departure from   standard systems. 

This paper continues, from a new perspective, the line of investigation of the above mentioned papers.  Our work is also motivated by more geometric considerations. The dynamics of geodesic flows on manifolds with discontinuous Riemannian metrics is a potentially interesting topic that, to our knowledge, is still waiting for a detailed study. When the metric tensor $g$ is in the conformal class of a smooth metric $g_0$, that is,  $g=\eta^2 g_0$ for   a positive function $\eta$  with discontinuities on smooth hypersurfaces,  we have a system very similar to  our lensed billiard systems.  For more general discontinuous metrics, Snell's law still holds as shown in \cite{GG}. It is hoped that some of the observations about the dynamics of lensed billiards in dimension $2$ will serve as a guide into the dynamics of such geodesic flows. (Having   this Riemannian setting in mind, some of the details relegated to appendices are given in greater generality than needed for the narrower purposes of the present paper.)

The rest of the paper is organized as follows.
Lensed billiard systems are introduced precisely in Section \ref{lensedbilliards}. In defining their phase space, a certain care is needed when accounting for refractions. While the systems we study can be considered as a special case of soft billiards with discontinuous potentials, the perspective we take in defining the billiard map has a number of benefits related to the issue of handling refractions and the differential of the billiard map when refractions occur.
Just as for ordinary billiard systems, lensed billiards are dynamical systems with singularities. Besides those singularities typically present in ordinary billiards (hitting a corner or grazing trajectories),
 we should add singular sets caused by the {\em critical angle} of incidence, a quantity to be defined formally in Subsection \ref{lb} that specifies the transition between  reflected and refracted trajectories.
Subsection \ref{lb} defines the billiard map of lensed billiards, introduces the notion of the critical angle, and demonstrates the singular sets and phase portraits for a collection of examples to be studied later in the paper.
In Subsection \ref{switching} we make a few  observations concerning refractions based on the invariance of the  Liouville measure  under the lensed billiard map. As a first step in the analysis of lensed billiard dynamics, we describe these systems as a switching process involving  two ordinary
open billiard systems and obtain mean values for time and number of collisions during sojourns in each subsystem.   The idea of  switching dynamics and the statistics of {\em sojourn times} discussed in Section \ref{switching} closely links the analysis of lensed billiards with the study of open (standard) billiard systems. (See Figure \ref{density_sojourn} and the remarks made around it.)
In Subsection \ref{R2} we describe the differential of the lensed billiard map. This is used afterwards to obtain numerically Lyapunov exponents for the examples of parametric families of billiard systems discussed in Section \ref{observations}, which contains the main results of the paper. The central interest  is to explore the Lyapunov exponent's dependence on the potential parameter $C$.  
 We identify a few properties of lensed billiards that recur in the examples and use the results of Subsection \ref{switching}  to explain some of them.
 
 The authors wish to thank Zhenyi Zhang for pointing out several corrections to an early draft of the paper.

\section{Preliminary details and facts}\label{lensedbilliards}
 
\subsection{Definition of lensed billiards}\label{lb}
Let  $\mathcal{R} \subseteq \mathbb{R}^2$ be a closed set
 with piecewise smooth boundary and $\mathcal{A}$  an open subset of
$\mathcal{R}$ whose boundary is also piecewise smooth.  Part of the boundary of $\mathcal{A}$ may be contained in the boundary of $\mathcal{R}$. Let the potential function be $V=C\mathbbm{1}_\mathcal{A}$ for a given constant $C \in \mathbb{R}$. The {\em lensed billiard  map} (and  billiard flow) is defined on the phase space of the system, which involves a small modification of the standard billiard phase space definition.
 \begin{figure}[htbp]
\begin{center}
\includegraphics[width=2.5in]{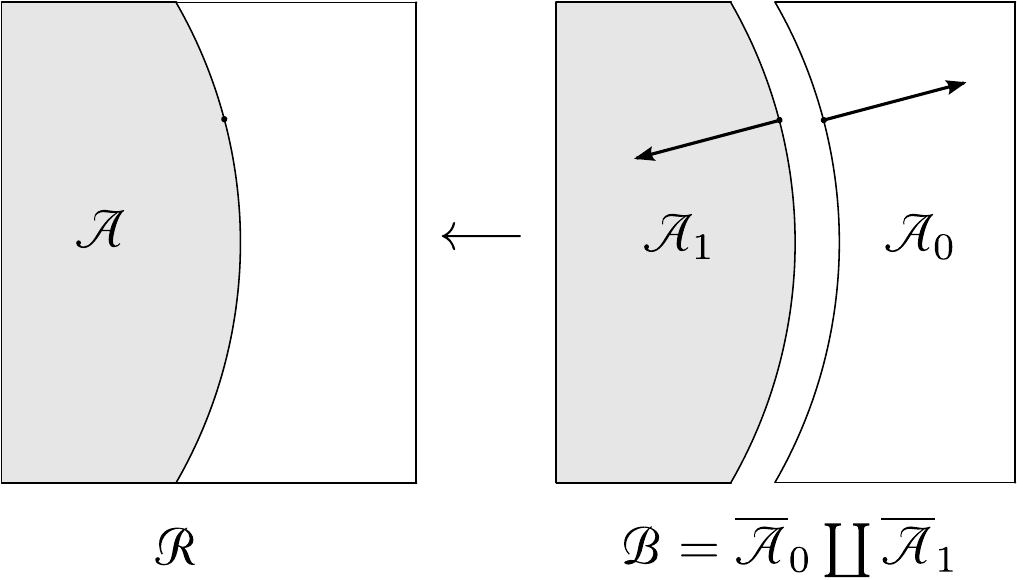}\ \ 
\caption{{\small In the  lensed billiard phase space the  refracting boundary is duplicated. }}
\label{shape2}
\end{center}
\end{figure}  

 Let  $\mathcal{A}_1:=\mathcal{A}$ and $\mathcal{A}_0:=\mathcal{R}\setminus\overline{\mathcal{A}}$ and consider the disjoint union $\mathcal{B}:=\overline{\mathcal{A}}_0\coprod\overline{\mathcal{A}}_1$.  We fix a value $E>C$ for the total energy (kinetic plus potential). By a {\em regular point} in the boundary  $\partial \mathcal{B}$ we mean
 a point $x$ at which the unit vector $\mathbf{n}_x$ perpendicular to  $\partial \mathcal{B}$ and pointing towards the interior of $\mathcal{B}$ at $x$ is defined. Note that 
 a piece of boundary common to both $\mathcal{A}_0$ and $\mathcal{A}_1$ is accounted for twice.  These  hypersurfaces of discontinuity of the potential
 function will have a copy as part of the boundary of $\mathcal{A}_0$ and another as part of the boundary of $\mathcal{A}_1$. Since these are sets where refracting is possible, we say that they are  contained in the {\em refracting boundary}. The others that are not duplicated are part of the {\em reflecting boundary}.

 The phase space $\mathcal{V}$ of the billiard map is the set of pairs $(x,v)$ where $x$ is a regular point in the boundary of $\mathcal{B}$ and $v$ is a tangent vector at $x$
 such that $\langle \mathbf{n}_x, v\rangle \geq 0$ and $$|v|=\begin{cases}\sqrt{2E/m} & \text{ if } x\in \overline{\mathcal{A}}_0\\
 \sqrt{2(E-C)/m} & \text{ if } x\in \overline{\mathcal{A}}_1.
 \end{cases}
 $$
 We define the billiard map $$\mathcal{T}(x,v)=(X(x,v), V(x,v)), \ \ (x,v)\in \mathcal{V},$$  as follows.  Suppose $x\in \overline{\mathcal{A}}_i$. Let $t_0:=\inf\{t>0: x+tv\in \partial\mathcal{B}\}.$
 If   $y=\gamma(t_0)$ is not a regular boundary point,
 the billiard map is not defined. If it is regular, let $y_j$, $j=0,1$, be the corresponding point in $\overline{\mathcal{A}}_j$. Set  $\tilde{v}:=\gamma'(t_0)$. We now determine whether the new velocity should be defined by a reflection or a refraction of $\tilde{v}$. 
 We say that 
 the angle $\theta$ between
 $\tilde{v}$ and $\mathbf{n}_{\gamma(t_0)}$ is {\em critical} if the normal component $\tilde{v}_n$ of $\tilde{v}$ satisfies 
 $$ \left(\tilde{v}_n\right)^2=\frac2m(C_j-C_i).$$
 Equivalently, this happens for the angle $\theta_{\text{\tiny crit}}$ that satisfies:
 \begin{equation}\label{Critical}\sin \theta_{\text{\tiny crit}}=\sqrt{\frac{E-C_j}{E-C_i}}. \end{equation}
 Note that a critical angle of incidence is only possible when $C_i<C_j$; that is, when approaching a region of higher potential from a region with lower potential. The state corresponding to a critical angle of approach is  then included in the singular set of the billiard map and the trajectory is terminated.
 If the velocity vector undergoes a reflection, we set $X(x,v)=y_i$ (recall  $x\in \mathcal{A}_i$);  if it undergoes a refraction, then $X(x,v)=y_j$, $j\neq i$.
 
  \begin{figure}[htbp]
\begin{center}
\includegraphics[width=2.0in]{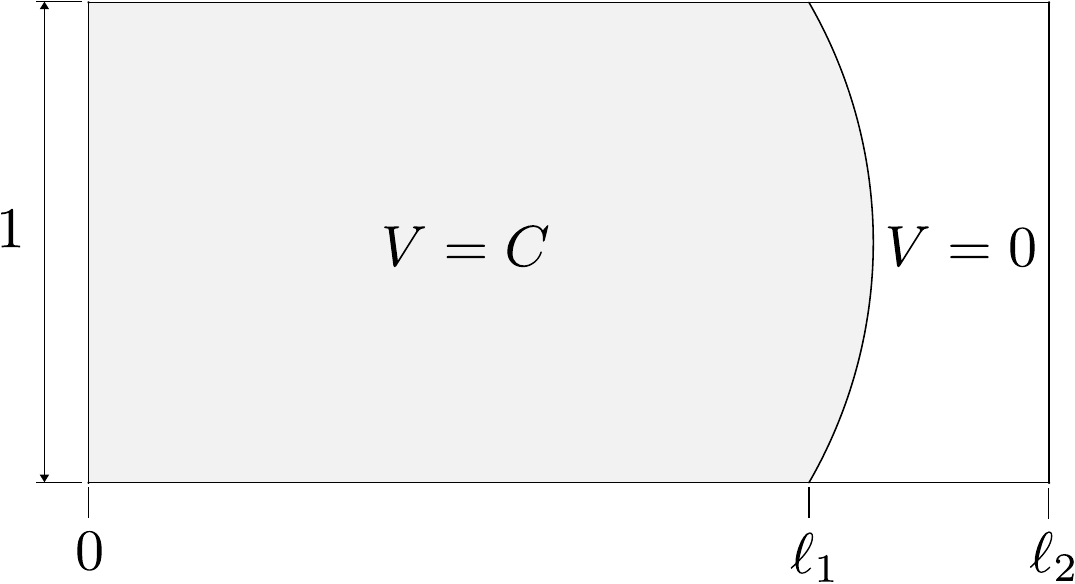}\ \ 
\caption{{\small This lensed billiard table involves $4$ parameters:  $\ell_1, \ell_2$, the signed curvature $\kappa\in [-2,2]$ of the
interface segment, and the value $C$ of the potential function to the left of the interface segment. To its right the potential is $0$. The figure shows 
one example for $\kappa<0$. Initial conditions for billiard trajectories are set at the vertical wall on the right. If $C$ is sufficiently large ($C\geq 1$ if the initial particle speed is set
equal to $\sqrt{2}$) for the system to be a standard billiard on the region to the right of the interface segment, then $\kappa<0$ corresponds to a semidispersing billiard and
$\kappa>0$ to a focusing billiard. }}
\label{shape1}
\end{center}
\end{figure}  

 In order to define $V(x,v)$, let  $\tilde{v}=\tilde{v}_\tau+\tilde{v}_n\mathbf{n}_{\gamma(t_0)}$
 be the orthogonal decomposition of $\tilde{v}$ at $\gamma(t_0)$. 
 The tangential component of the new velocity  in both cases is
 $$V_\tau(x,v) = \tilde{v}_\tau $$
 and the normal component is
 $$
  V_n(x,v) =
 \begin{cases}
-\tilde{v}_n & \text{ if } \left(\tilde{v}_n\right)^2-\frac2m (C_j-C_i)<0\ \ \ \ (\text{reflection})\\[0.5em]
 \sqrt{\left(\tilde{v}_n\right)^2 -\frac{2(C_j-C_i)}{m}}& \text{ if } \left(\tilde{v}_n\right)^2-\frac2m (C_j-C_i)>0\ \ \ \ (\text{refraction})\\[0.5em]
 0 &\text{ if } \text{ incidence angle is critical}. 
 \end{cases}  $$

\begin{figure}[htbp]
\begin{center}
\includegraphics[width=5.3in]{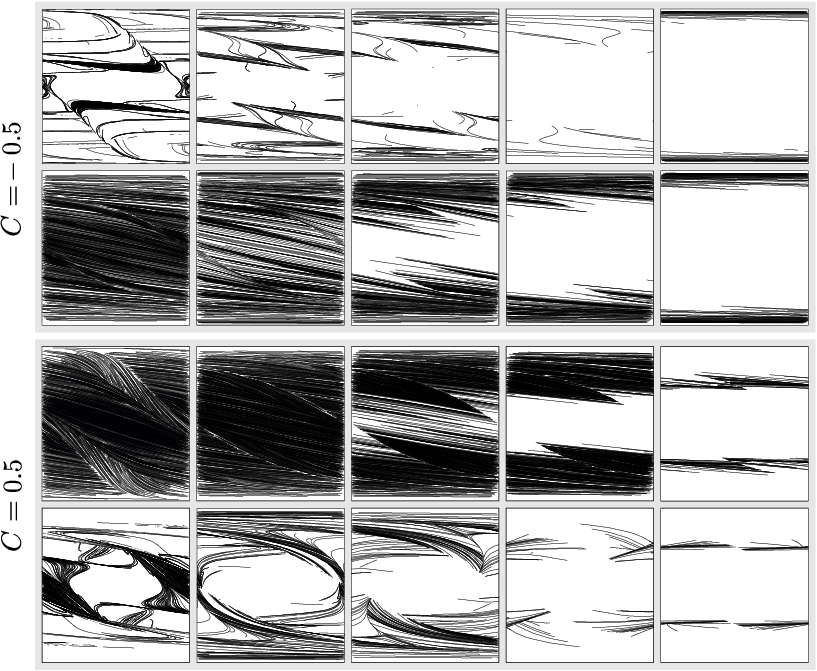}\ \ 
\caption{{\small Singular set  for the  family of Figure \ref{shape1}, restricted to the part of the phase space over the right vertical side of the rectangle.
 The parameters are $\ell_1=0.5$,
$\ell_2=1$, $C=-0.5$ (top group of ten plots), $C=0.5$ (bottom group); in each group of plots, the values of $\kappa$ are, clockwise from the top left: $-2.0$, $-1.5$, $-1.0$,
$-0.5$, $-0.1$, $0.1$, $0.5$, $1.0$, $1.5$, $2.0$. Only the singular initial conditions leading to the critical refracting angle are shown, for orbits of length 
less than $16$.}}
\label{SingKC}
\end{center}
\end{figure}

 This rule for obtaining $V_n$ is easily justified, for example, by introducing a smooth potential transition over a narrow tubular neighborhood of the line of discontinuity of width $\epsilon>0$, then taking the limit of the solution to Newton's equation of motion as  $\epsilon$ approaches $0$. We relegate the details (of a more general fact) to  Appendix  \ref{Ap_A}.  
 
It would be possible, under a convexity assumption, to define the billiard map when the angle of incidence is critical, but we choose to exclude states leading to such angles from the
domain of definition of the billiard map.  States in $\mathcal{V}$ whose orbits generated by  the map $\mathcal{T}$, after a finite number of steps, lead to a critical 
 angle or to a non-regular point  of $\partial \mathcal{B}$, or to tangential contact, will be called {\em singular}. 
The billiard table in Figure \ref{shape1} will be used as a multiparameter  family of examples at a number of places. For this example, 
Figure \ref{SingKC} illustrates   the  points in the singular set associated to critical angles. These are the new critical points not present in standard billiards.

  \begin{figure}[htbp]
\begin{center}
\includegraphics[width=4.9in]{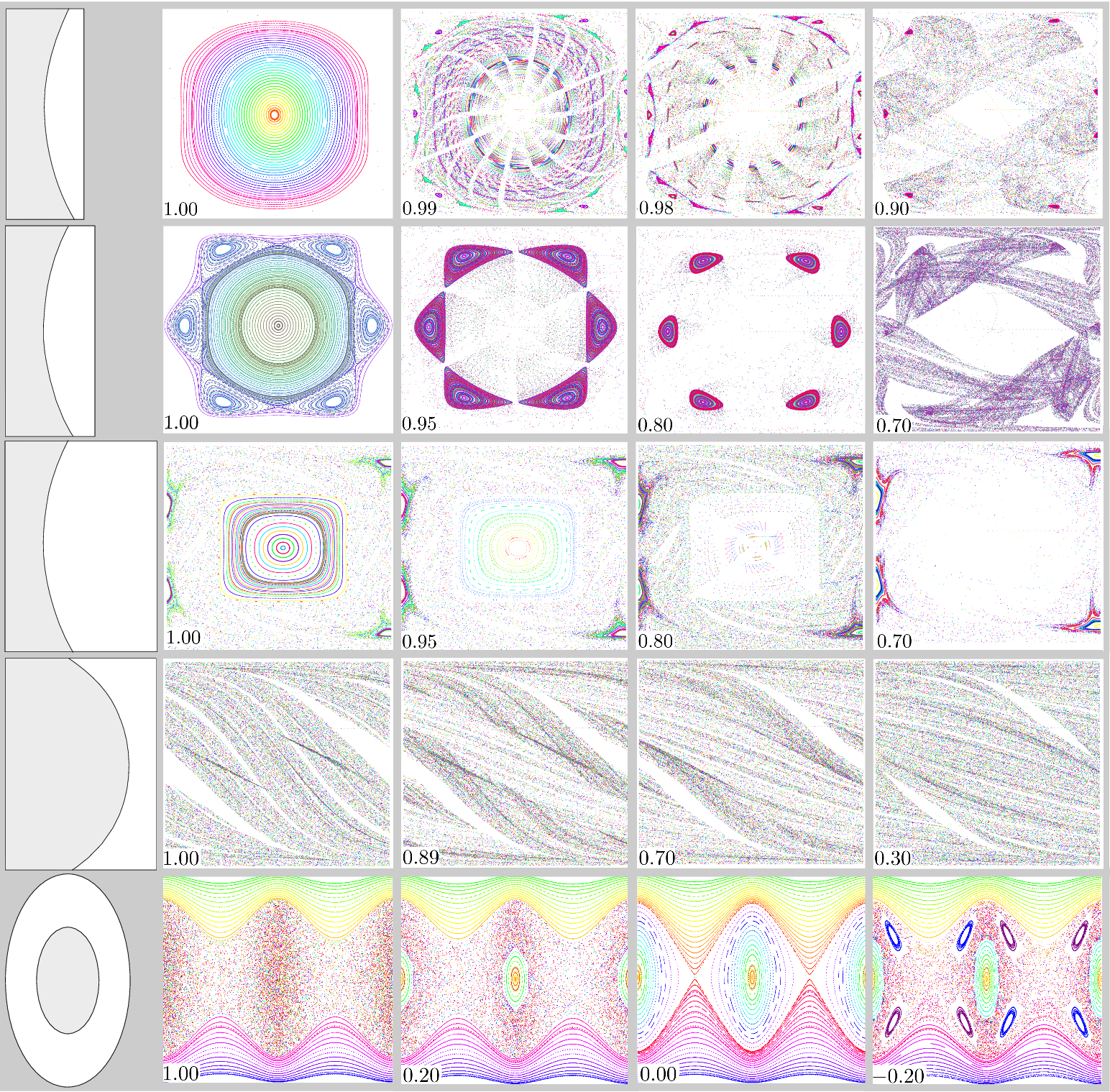}\ \ 
\caption{{\small Phase-portraits for bounded lensed billiard systems. Numbers indicate values of the potential constant. The orbits shown here are  for the return map to the right vertical side of the rectangular domain, for the top four systems, and the return map to the outer ellipse for the bottom system, rather than $\mathcal{T}$ itself. Only the relevant part of the phase space is shown.}}
\label{phase_portrait}
\end{center}
\end{figure}  
The figure only shows the part of the phase space over the vertical segments on the right end of the rectangular domain. Further details are given in the figure's caption.
Figure \ref{phase_portrait} shows a few phase portraits of   examples taken from the family of Figure \ref{shape1}. The plots only show orbits of the return map to the right vertical
side of the rectangular domain.

 It will be assumed  (as it is  often done for standard billiard systems) that the domain of $\mathcal{T}$ in $\mathcal{V}$ is 
an open set of  full Lebesgue measure and that $\mathcal{T}$ is a smooth map there. This  will be valid for the examples in this paper.

\subsection{Switching dynamics and measure invariance}\label{switching}
A prominent aspect of the dynamics of lensed billiards is the back-and-forth switching process between two standard billiard systems in
$\overline{\mathcal{A}}_0$ and $\overline{\mathcal{A}}_1$.  Segments of trajectories   from the moment of arrival in region $\overline{\mathcal{A}}_i$ to the moment of the
next switch back into the other region will be called  {\em sojourns}. The statistics of time and number of collisions during a sojourn   plays a role in
the analysis of lensed billiards, as will be seen. We give in this section the mean value of these two quantities under the assumption that the two standard  billiard subsystems are ergodic.

It is useful throughout to keep in mind the mechanical counterpart of {\em Snell's law.} Recall from the definition of the refraction map given earlier that, upon refraction at a point $x$ of discontinuity of the potential, the  components of the incident and transmitted velocities tangent to the line of discontinuity are equal.  If the trajectory leaves $\mathcal{A}_i$ and enters $\mathcal{A}_j$, $i\neq j$, and if $\theta_i$ and $\theta_j$ are the angles of the velocities $v_i$ and $v_j$ pre- and post-refraction, respectively, relative to the unit normal vector $\mathbf{n}_x$ to the line of discontinuity pointing into $\mathcal{A}_j$, then equality of tangential velocity components amounts to
$|v_i|\sin\theta_i =|v_j|\sin\theta_j.$ It follows that
$$ \frac{\sin\theta_i}{\sin\theta_j}=\frac{|v_j|}{|v_i|}=\sqrt{\frac{\frac12m|v_j|^2}{\frac12m|v_i|^2}}=\sqrt{\frac{E-C_j}{E-C_i}}.$$
This is Snell's law, in which terms of the form $\sqrt{E-C}$ assume the role of the refractive index in geometric optics.

A few remarks will be needed concerning invariance of the canonical billiard measure. If $V$ is the potential function on $\mathcal{R}$ and $\mathcal{V}$ is the billiard phase space, let $\mathcal{V}_E$ consist of the points $(x,v)$ in $\mathcal{V}$ such that $|v|=\sqrt{2(E-V(x))/m}$. This
is the invariant subset of states with total energy $E$.  Denoting by $\theta$ the angle between $\mathbf{n}_x$ and $v$, and by $s$ the arclength parameter along boundary line segments, the
{\em canonical billiard measure}, or {\em Liouville measure}, is the measure $\nu$ on $\mathcal{V}_E$ such that
\begin{equation}\label{Liouville}d\nu(s,\theta) =\sqrt{2(E-V(x))}\cos\theta\, d\theta\, ds.\end{equation}
It is well-known that this measure is preserved by the billiard map. That the lensed billiard map also preserves this measure is briefly justified in Appendix \ref{Ap_B}.  We refer to the angle distribution in $\nu$ as the {\em cosine law}.

We note the following effect of the square root term appearing in Equation (\ref{Liouville}).  The invariant measure on the part of the lensed billiard phase space
corresponding to an interface segment  between regions with potentials $C_0$ and $C_1$ is proportional to $\sqrt{E-C_0}$ on the side
of potential value $C_0$ and $\sqrt{E-C_1}$ on the other side. Thus in the long run, for an ergodic  system, the ratio of probabilities of transitions 
into $C_0$  (from either $C_0$ or $C_1$)
  over transitions into $C_1$ is
$$p_{01}=\sqrt{\frac{E-C_0}{E-C_1}}. $$

Let us illustrate this point   using  the lensed variant of the Bunimovich  stadium  shown in Figure \ref{BuniMeasure}, which may be expected to be ergodic, although we do not
verify it here.

\begin{figure}[htbp]
\begin{center}
\includegraphics[width=2in]{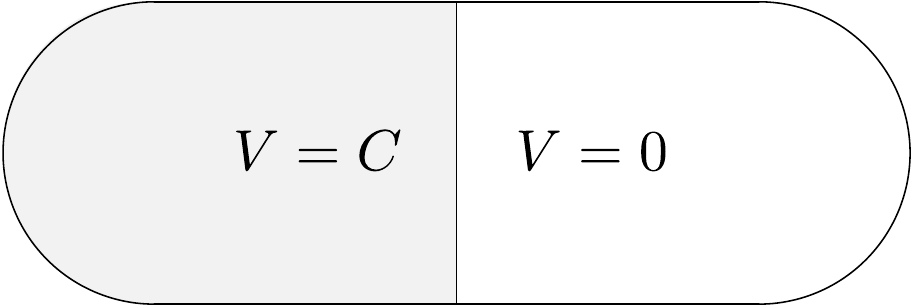}\ \ 
\caption{{\small   Billiard used to illustrate crossing probabilities.}}
\label{BuniMeasure}
\end{center}
\end{figure} 

Let us focus on the vertical interface line in the middle of the table. The probability that the particle, upon reaching the vertical line (from the left or from the right),
will continue towards the right by either reflection or refraction
   ($p_+$) over the probability of continuing  towards the left ($p_-$) should be (taking mass $m=1$ and energy $E=1$)
$$\frac{p+}{p_-}= \frac{1}{\sqrt{1-C}}.$$ 
In particular, when $C$ is negative and $|C|$ is very large, this probability is close to $0$. This means that once the particle falls into the region of negative potential,
most ``collisions'' with the vertical line correspond to reflections on the left side. The opposite holds when $C$ lies between $0$ and $1$ and close to $1$.
As an illustration, the following table gives predicted values $p_+/p_-$ and values obtained by numerical simulation for a randomly chosen orbit of $500000$ steps. 
$$ 
\begin{array}{c|c|c}C & \text{approx. ratio} & \frac{p_+}{p_-} \\\hline 1/3 & 1.222 & 1.225 \\1/2 & 1.418 & 1.414 \\8/9 & 2.998 & 3.000 \\-8 & 0.336 & 0.333\end{array}
$$
Comparing two boundary segments in a region with same value for the potential, 
the ratio of the number of crossings of the vertical segment from right to left over the number of collisions with the horizontal  segment of the same length in the region where
the potential is $C$ (chosen equal to $-1$ in the simulation, for a trajectory of $500000$ steps),  was $0.987$,  to be compared with the exact value $1$.

Another manifestation of measure invariance is recorded in the following proposition. (See Appendix \ref{Ap_B} for further remarks.)
\begin{proposition}\label{crosscos}
Suppose $C_0>C_1$. Let $x\in\mathcal{C}:= \overline{\mathcal{A}}_0\cap\overline{\mathcal{A}}_1$ be a regular point and $\mathbf{n}_x$ the unit normal vector to $\mathcal{C}$ at $x$ pointing into $\mathcal{A}_0$.   Further suppose that the velocity  of trajectories incident upon $\mathcal{C}$ at $x$ coming from  $\mathcal{A}_1$ is distributed according to the cosine law. Then the  distribution of post-crossing velocities, conditional on the occurrence of  refraction, also satisfies the cosine law. 
\end{proposition}

A standard application of the Ergodic Theorem (see Appendix \ref{Ap_C} for more details) gives the next proposition, for which we use the following notation. Let $\mathcal{A}$ be either $\mathcal{A}_0$ or $\mathcal{A}_1$. Let $A$ and $L$ be the area and the boundary length of $\mathcal{A}$, respectively, and $\mathcal{C}:= \overline{\mathcal{A}}_0\cap\overline{\mathcal{A}}_1$. Denote by $\ell$ the length of $\mathcal{C}$, by $L$ the length of the boundary of $\mathcal{A}$ and by $A$ the area of $\mathcal{A}$. Let $\mathcal{V}$ be the space of pairs $(x,v)$ where $x\in \partial \overline{\mathcal{A}}$ and $v$ is a  tangent vector  to $\overline{\mathcal{A}}$ at $x$ pointing into $\mathcal{A}$ and having norm  $\mathcal{s}$. Similarly, we denote by $\mathcal{E}$ the space of pairs $(x,v)$ where  now $x\in \mathcal{C}$. 
 At each $(x,v)\in \mathcal{E}$, let  $T(x,v)$ and $N(x,v)$  denote, respectively,  the time of first return to $\mathcal{E}$ and the number of collisions  with the boundary of $\overline{\mathcal{A}}$
  of a billiard trajectory with initial state $(x,v)$
  before returning to $\mathcal{E}$. For each $(x,v)\in \mathcal{V}$, let $\tau(x,v)$ denote the time duration of free flight from $(x,v)$  to the point of next collision. This is
 naturally the length of the free flight divided by the speed $\mathcal{s}$. Finally, denoting by $\nu$ and $\nu_\mathcal{E}$ the normalized Liouville measure on $\mathcal{V}$ and
 $\mathcal{E}$, respectively, we introduce the mean values
$$
 \langle N\rangle_{\mathcal{E}} :=\int_{\mathcal{E}} N(x,v)\, d\nu(x,v), \ \  
 \langle T\rangle_{\mathcal{E}}:=\int_{\mathcal{E}} T(x,v)\, d\nu(x,v), \ \  
 \langle \tau\rangle_{\mathcal{V}} :=\int_{\mathcal{V}} \tau(x,v)\, d\nu(x,v).
$$

The following is a special case of Theorem \ref{averages} proved in Appendix \ref{Ap_C}.
\begin{proposition}\label{averages_2} With the notations just introduced and 
under the assumption that the standard billiard map in $\overline{\mathcal{A}}$ is ergodic, the following relations hold:
\begin{enumerate}
\item $\langle N\rangle_{\mathcal{E}} = \frac{L}\ell$;
\item $\langle T\rangle_{\mathcal{E}}= \frac{\pi A}{\ell \, \mathcal{s}}$;
\item $\langle \tau\rangle_{\mathcal{V}}=  \frac{\pi A}{L\mathcal{s}};$
\item $\langle T\rangle_{\mathcal{E}}=\langle N\rangle_{\mathcal{E}}\langle \tau\rangle_{\mathcal{V}}.$
\end{enumerate}
Furthermore, if $\mathcal{E}_0\subseteq \mathcal{E}$ consists of pairs $(x,v)$ such that $x\in \mathcal{C}$ and the angle $\theta$ which $v$ makes with the normal $\mathbf{n}_x$ satisfies
$|\sin\theta|<r_0$ then the mean  number of returns to $\mathcal{E}$ before a first return to $\mathcal{E}_0$ is
$1/r_0$.
\end{proposition}

Let us assume for concreteness that $C_0>C_1$. Transitions from $\mathcal{A}_0$ to $\mathcal{A}_1$ happen whenever the trajectory reaches the separating hypersurface, at any angle of incidence. On the other hand, transitions from $\mathcal{A}_1$ to $\mathcal{A}_0$ can only occur when the angle of incidence between the velocity vector and the normal to the 
separation line (pointing into $\mathcal{A}_0$) is smaller than the critical angle.  Denoting by $\theta$ the angle of incidence,  a transition from $\mathcal{A}_1$ to $\mathcal{A}_0$ happens when
$|\sin\theta|<r_0:=\sqrt{\frac{E-C_0}{E-C_1}}$.    (See Equation (\ref{Critical}).)
Furthermore, by Snell's law, the angle $\varphi$ that the trajectory makes with the same normal vector as it crosses into $\mathcal{A}_0$ satisfies
\begin{equation}\label{phi_theta} \sin\varphi = r_0^{-1}{\sin\theta}.\end{equation}

\begin{corollary}\label{sojourn stat}
Suppose the value of the potential function in $\mathcal{A}_i$ is $C_i$ with $C_0>C_1$ and that the standard billiard system in $\overline{\mathcal{A}}_i$
is ergodic for $i=0$ and $1$.  Let $\langle T_i\rangle$ and $\langle N_i\rangle$ denote the mean time and number of collisions of the lensed billiard system during a sojourn 
in $\mathcal{A}_i$ before the next switch to the other region.
Then
$$\frac{\langle T_0\rangle}{\langle T_1\rangle}=\frac{A_0}{A_1}, \ \ \frac{\langle N_0\rangle}{\langle N_1\rangle}=\frac{L_0}{L_1}r_0, $$
where $A_i$ and $L_i$ are the area  of $\mathcal{A}_i$ and the length of the boundary of $\mathcal{A}_i$.
\end{corollary}

\begin{figure}[htbp]
\begin{center}
\includegraphics[width=3.5in]{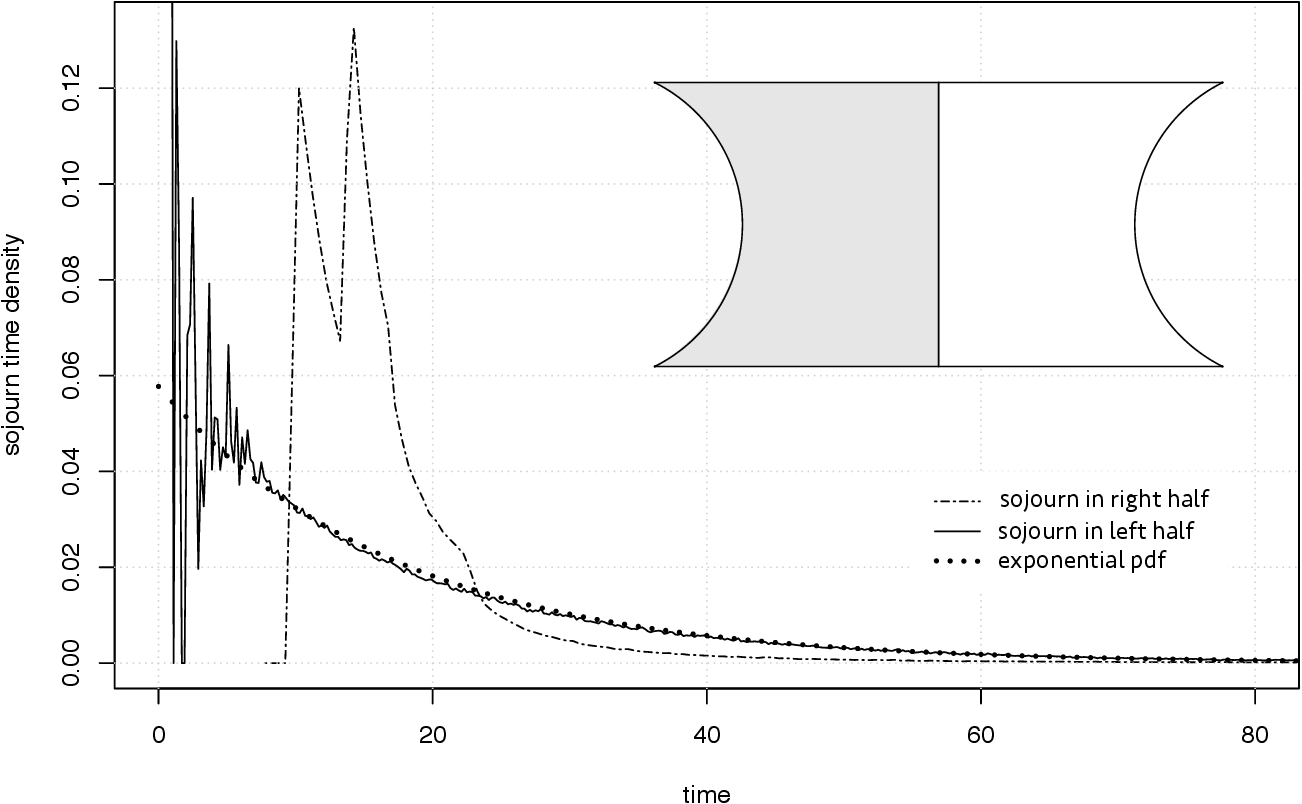}\ \ 
\caption{{\small  Sojourn time distributions obtained by simulating long trajectories of the lensed billiard system for the table shown in the inset. The potential is $-1$ on the left half of the table   and $0$ on the right, while particle energy is $0.01$. The particle undergoes many more collisions in the left than the right region, although   mean sojourn times are the same in  both. Theorem \ref{averages} gives the mean value $\approx 17.3054$ (the table has area $\approx 1.558$ and height $1$);  simulated values are $17.317$ (left) and $17.305$ (right).  The dotted line is the graph of the probability density function of the exponential distribution with parameter $\lambda=1/17.305=0.058$.}} 
\label{density_sojourn}
\end{center}
\end{figure}

Further  information about the probability distribution of sojourn times, beyond simply their mean values, is a focus of future study. 
For an example of what can be expected, consider
Figure \ref{density_sojourn}.  It shows the sojourn time distributions for the lensed billiard system whose billiard domain is shown in the figure inset. 
The left half of the table has potential $-1$ and the right half has potential $0$, while the total energy of the billiard particle is a small positive value. Thus the particle
undergoes many more collisions in a typical sojourn in the left-hand region. In order to cross back into the right-hand region, the billiard trajectory has to pass through a window
in phase space defined by the vertical middle line and a small angle interval. 
 It is interesting to note, in particular, the tail behavior for  the sojourn time distribution in the left  half, which is well approximated by the probability density of an exponential random
 variable with parameter equal to the reciprocal of the mean sojourn time.

It is likely that methods for the study of open billiards such as in \cite{BuniYao} and \cite{PeneSau} can be used for a detailed analysis of sojourn time distribution in potential wells. 
In lensed billiards  we are presented with two ordinary  billiard systems that are open to each other in the sense that the phase space of each  contains a subset, or hole, through which the billiard particle can
cross into the other. The size of the hole is smaller for the region with smaller potential. One should expect the successive hitting times, appropriately normalized, into small holes to be approximated by a Poisson process for hyperbolic billiard subsystems. 

\subsection{Differential of the lensed billiard map}\label{R2}
We wish here to obtain the differential of the lensed billiard map to be used in  the next section for the evaluation of Lyapunov exponents. 
For this purpose we introduce  
coordinates that differ in minor ways from  the more commonly employed conventions. The main difference is that, instead of the Jacobi coordinates adapted to wavefronts in a neighborhood of $(x,v)$ as in the standard reference \cite{chernov}, we use  arclength parameter along the boundary curves $\mathcal{S}$ of $\mathcal{A}_i$ and angles that are measured relative to
$v$ itself. The differential of the billiard map at reflections will differ from the more standard description by the absence of a cosine term at certain places. 
This small deviation from  standard conventions has been made in order to avoid having to  elaborate  on the behavior of wavefronts at refractions. For standard billiards, time and arclength parameters are everywhere proportional since particle speed is constant, but this is no longer the case across refractions, which introduces a few issues that we choose to avoid.

  Recall that the smooth pieces of the boundary of $\mathcal{B}$ are the pieces of the
boundaries of $\overline{\mathcal{A}}_0$ and $\overline{\mathcal{A}}_1$ oriented so that  the unit normal vector field $\mathbf{n}$  points towards the interior of $\mathcal{B}$. 
 Let $\mathbf{t}$ be the unit  tangent vector field to the boundary chosen so that $(\mathbf{t},\mathbf{n})$ is a positive orthogonal basis at each regular boundary point.
 Let us consider one step of the billiard trajectory from $\mathcal{O}_1$ in boundary piece $\mathcal{S}_1$ to $\mathcal{O}_2$ in  boundary piece $\mathcal{S}_2$. 
 The initial velocity
 is $v_1$ such that $v_1\cdot \mathbf{n}_1>0$ (standard dot product) and $\mathcal{T}(\mathcal{O}_1,v_1)=(\mathcal{O}_2,v_2)$. This is shown in the diagram of Figure \ref{differential} where, on the left, $\mathcal{T}$ produces a refraction and on the right a reflection.

\begin{figure}[htbp]
\begin{center}
\includegraphics[width=5.0in]{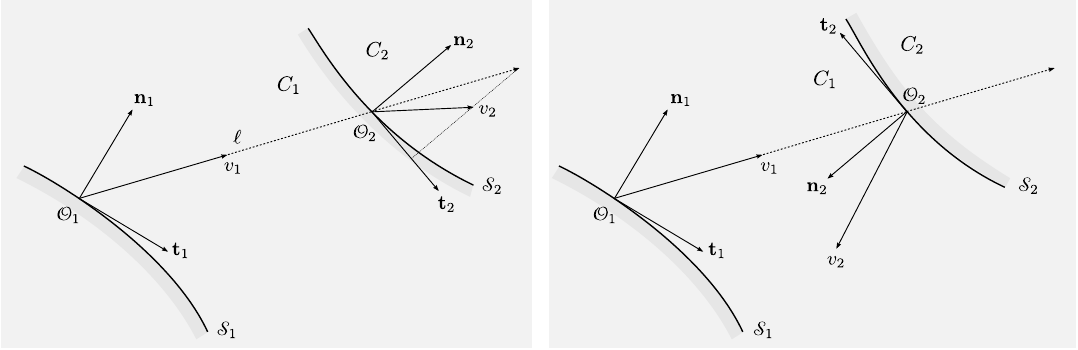}\ \ 
\caption{{\small Notation for the calculation of the differential of the lensed billiard. On the left a refraction and on the right a reflection. $C_1$ and $C_2$ are the values of the potential function at the indicated regions.  }}
\label{differential}
\end{center}
\end{figure} 
 
Let $J$ denote the  rotation matrix in $\mathbb{R}^2$ by $\pi/2$ counterclockwise. This is the generator of plane rotations: the linear map $v\mapsto \exp(\theta J)v$ rotates $v$ counterclockwise by angle $\theta$. Thus for a small interval of angles centered at $0$, this defines a neighborhood of velocities centered at $v$.
 We define coordinates $(x_1, \theta_1)$ in a neighborhood of $(\mathcal{O}_1, v_1)$ and  $(x_2, \theta_2)$ in a neighborhood of $(\mathcal{O}_2,v_2)$ so that any $(q,v)$ in the first neighborhood  and $(Q,V)=\mathcal{T}(q,v)$ in the second satisfy:
\begin{equation}\label{qQvV}
\begin{cases}
q&=\ \gamma_1(x_1) \text{ is a local parametrization of $\mathcal{S}_1$ by  arclength  with $\gamma_1(0)=\mathcal{O}_1$},\\
Q&=\ \gamma_2(x_2) \text{   is a local parametrization of $\mathcal{S}_2$ by arclength with $\gamma_2(0)=\mathcal{O}_2$},\\
v&=\ \exp(\theta_1 J)v_1,\\
V&=\ \exp(\theta_2J)v_2.
\end{cases}
\end{equation}

We wish to obtain the differential of $\mathcal{T}$ at $(\mathcal{O}_1,v_1)$:
$$ d\mathcal{T}_{(\mathcal{O}_1,v_1)}=\left(\begin{array}{cc}\frac{\partial x_2}{\partial x_1}(0,0) & \frac{\partial x_2}{\partial \theta_1}(0,0) \\[0.5em] \frac{\partial \theta_2}{\partial x_1}(0,0) & \frac{\partial \theta_2}{\partial\theta_1}(0,0)\end{array}\right).$$
The following additional notation is needed. Let $\kappa$ denote the geodesic curvature of  $\mathcal{S}_i$, which is defined by
$$ D_{\mathbf{t}}\mathbf{t} = \kappa \mathbf{n},$$
where $D_\mathbf{t}$ indicates standard directional  derivative of vector fields along $\mathbf{t}$.
  Let  
 $\ell$ be the Euclidean distance between $\mathcal{O}_1$ and $\mathcal{O}_2$, and set $\nu_i=v_i/|v_i|$, $i=1,2$. Further write $E$ for the total energy (kinetic plus potential) and
 $C_1$, $C_2$ for the values of the potential at the regions indicated in Figure \ref{differential}.
\begin{theorem}[Differential of the lensed billiard map]\label{billiard differential}
With the notations just given, the differential $d\mathcal{T}_{(\mathcal{O}_1,v_1)}$ in the coordinate systems $(x_i,\theta_i)$ is given by
$$ d\mathcal{T}_{(\mathcal{O}_1,v_1)}=\left(\begin{array}{cc}\frac{\nu_1\cdot\mathbf{n}_1(\mathcal{O}_1)}{\nu_1\cdot\mathbf{n}_2(\mathcal{O}_2)} & -\frac{\ell}{\nu_1\cdot\mathbf{n}_2(\mathcal{O}_2)}\\[0.5em] 
-2\kappa(\mathcal{O}_2)\frac{\nu_1\cdot \mathbf{n}_1(\mathcal{O}_1)}{\nu_2\cdot\mathbf{n}_2(\mathcal{O}_2)} & \left[1+\frac{2\ell \kappa(\mathcal{O}_2)}{\nu_1\cdot\mathbf{n}_2(\mathcal{O}_2)}\right]\frac{\nu_1\cdot\mathbf{n}_2(\mathcal{O}_2)}{\nu_2\cdot\mathbf{n}_2(\mathcal{O}_2)}\end{array}\right)$$
if a reflection occurs at $\mathcal{O}_2$ or
$$ d\mathcal{T}_{(\mathcal{O}_1,v_1)}=\left(\begin{array}{cc}\frac{\nu_1\cdot\mathbf{n}_1(\mathcal{O}_1)}{\nu_1\cdot\mathbf{n}_2(\mathcal{O}_2)}  &  -\frac{\ell}{\nu_1\cdot\mathbf{n}_2(\mathcal{O}_2)} \\[0.5em]
-\alpha\kappa(\mathcal{O}_2)\frac{\nu_1\cdot \mathbf{n}_1(\mathcal{O}_1)}{\nu_2\cdot\mathbf{n}_2(\mathcal{O}_2)}\sqrt{\frac{E-C_1}{E-C_2}}& \left[1+\frac{\alpha\ell\kappa(\mathcal{O}_2)}{\nu_1\cdot \mathbf{n}_2(\mathcal{O}_2)}\right]\frac{\nu_1\cdot\mathbf{n}_2(\mathcal{O}_2)}{\nu_2\cdot\mathbf{n}_2(\mathcal{O}_2)}\sqrt{\frac{E-C_1}{E-C_2}}\end{array}\right)$$
if a refraction occurs. In the latter case,
$$\alpha := 1-\left[
1-\frac{C_2-C_1}{E-C_1}\frac1{(\nu_1\cdot \mathbf{n}_2(\mathcal{O}_2))^2}
\right]^{\frac12}. $$
The determinant of  $d\mathcal{T}_{(\mathcal{O}_1,v_1)}$ is $v_1\cdot\mathbf{n}_1(\mathcal{O}_1)/v_2\cdot\mathbf{n}_2(\mathcal{O}_2)$ in both cases.
\end{theorem}

The similarities between the differentials for reflection and refraction become more apparent with the following  notation: $\ell_i:=\ell |v_i|$ where we recall that
 $|v_i|=\sqrt{2(E-C_i)/m}$. Then
\begin{align*}
 \left(d\mathcal{T}_{(\mathcal{O}_1,v_1)}\right)_\text{\tiny reflection}&=\left(\begin{array}{cc}\frac{v_1\cdot\mathbf{n}_1(\mathcal{O}_1)}{v_1\cdot\mathbf{n}_2(\mathcal{O}_2)} & -\frac{\ell_1}{v_1\cdot\mathbf{n}_2(\mathcal{O}_2)}\\[0.5em] 
-2\kappa(\mathcal{O}_2)\frac{v_1\cdot \mathbf{n}_1(\mathcal{O}_1)}{v_2\cdot\mathbf{n}_2(\mathcal{O}_2)} & \left[1+\frac{2\ell_1 \kappa(\mathcal{O}_2)}{v_1\cdot\mathbf{n}_2(\mathcal{O}_2)}\right]\frac{v_1\cdot\mathbf{n}_2(\mathcal{O}_2)}{v_2\cdot\mathbf{n}_2(\mathcal{O}_2)}\end{array}\right)\\
\left(d\mathcal{T}_{(\mathcal{O}_1,v_1)}\right)_{\text{\tiny refraction}}&=\left(\begin{array}{cc}\frac{v_1\cdot\mathbf{n}_1(\mathcal{O}_1)}{v_1\cdot\mathbf{n}_2(\mathcal{O}_2)}  &  -\frac{\ell_1}{v_1\cdot\mathbf{n}_2(\mathcal{O}_2)} \\[0.5em]
-\alpha\kappa(\mathcal{O}_2)\frac{v_1\cdot \mathbf{n}_1(\mathcal{O}_1)}{v_2\cdot\mathbf{n}_2(\mathcal{O}_2)}& \left[1+\frac{\alpha\ell_1\kappa(\mathcal{O}_2)}{v_1\cdot \mathbf{n}_2(\mathcal{O}_2)}\right]\frac{v_1\cdot\mathbf{n}_2(\mathcal{O}_2)}{v_2\cdot\mathbf{n}_2(\mathcal{O}_2)}\end{array}\right).
\end{align*}

Although the proof of Theorem \ref{billiard differential} is relatively straightforward, we nevertheless provide the details in Appendix \ref{Ap_D} since  refractions are not considered in standard references, and on account of the differences with the more commonly used coordinate choice for ordinary billiards. 

\section{Numerical study of Lyapunov exponents}\label{observations}

In this section we wish to explore, mainly numerically, the ways in which the presence of lenses affects the billiard dynamics as measured by the Lyapunov exponents of the billiard flow.   Thus we are concerned with the following quantity:
$$
\chi:=
\lim_{n\rightarrow \infty}\frac{ \log \left\|d\mathcal{T}_{(x,v)}^n \xi\right\|}{T_n},   
 $$
where $T_n$ is the time elapsed up to step $n$. 

Due to invariance of the Liouville measure (which is associated to an invariant $2$-form on phase space) we know that, on ergodic components, exponents come in pairs $\chi,-\chi$. We refer to the nonnegative  value as the {\em Lyapunov exponent} of the system.
In the numerical experiments shown here, we compute the average of $\chi$ over  a large sample of randomly chosen initial conditions.  
Not all the examples considered are ergodic, so the numerically computed $\chi$   may involve an average over the  values on ergodic components. 
Also note that exponents are obtained numerically using the differential of the billiard map described in Section \ref{R2}.

In all cases, we set the mass parameter $m=1$
and particle total energy  $E=1$. Thus particle  speed in regions where the potential is $0$ equals $\sqrt{2}$ and a region $\mathcal{A}$  on which the potential value is greater than or equal to $1$ acts as a standard billiard scatterer; that is, the boundary of $\mathcal{A}$ becomes part of the reflecting boundary.

We identify a number  of common features specific to lensed billiards among the examples considered here. We  also provide  conceptual explanations for some of these features  based on the results of Section
\ref{switching}. The  present discussion is exploratory in nature and the explanations for the observed features
 are heuristic and informal, although  we expect they will serve as a  basis for a more detailed study to  be pursued elsewhere.

\subsection{A multiparameter billiard case study}
Let us consider first the multiparameter family   depicted in  Figure \ref{shape1}. We wish to use it to illustrate a number of  properties that seem to be   typical of lensed billiards
but distinct in comparison with standard (purely reflecting) billiards.

As a first experiment,
let us obtain   approximate values for the positive Lyapunov exponent   as a function of the curvature parameter
$\kappa\in [-2,2]$ when $C=1$, initial speed is $\sqrt{2}$, $\ell_1=1$ and $\ell_2=2$.  Since $C=1$, this is a standard billiard system on the region
where the potential is $0$. This will serve as a basis for comparison when setting $C<1$.

\begin{figure}[htbp]
\begin{center}
\includegraphics[width=3.5in]{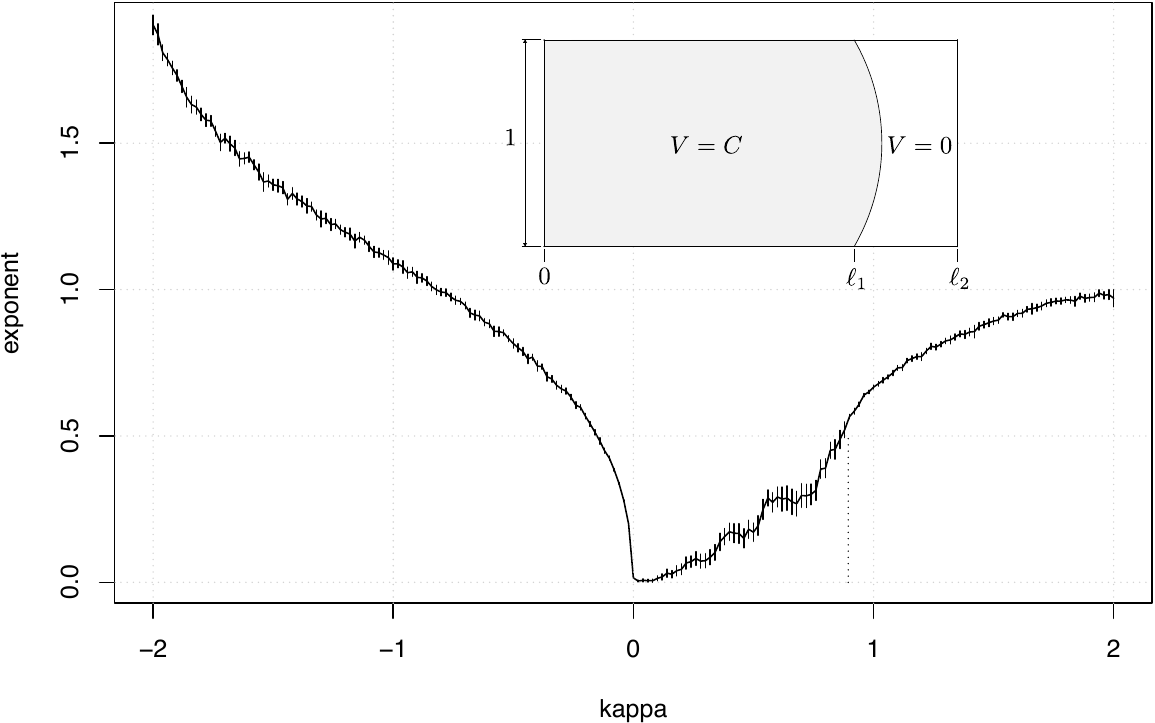}\ \ 
\caption{{\small Positive Lyapunov exponent as function of the signed curvature parameter $\kappa$.  The parameters are $\ell_2-\ell_1=1$ and $C=1$, so
this is a standard (purely reflecting) billiard in the complement of the shaded region (see insert), to the right of the curved line, where $V=0$.
}}
\label{Exper1_graph1}
\end{center}
\end{figure}  

The results are shown in Figure \ref{Exper1_graph1}. Negative values of $\kappa$ correspond to semidispersing billiards and positive values to focusing billiards. For  each of  $200$ equally spaced values of $\kappa$ we compute the mean value for the exponent over a sample of $100$ orbits of length $1000$ with random initial conditions: points
are taken from the vertical wall on the right with the uniform distribution and velocities according to the cosine law density. Whiskers represent $95\%$ confidence intervals based only on
the sample variation of the simulated data.  Other potential sources of errors may not be accounted for in this and similar  graphs.

It is worth noting  the  distinct nature of the graph over the interval $0<\kappa<2/\sqrt{5}$. (The dashed vertical line is at $\kappa=2/\sqrt{5}$.) The 
upper limit of this interval is the curvature for  which the     center of the circular arc lies  on the vertical right wall. The less well-defined behavior in the range $0\leq \kappa\leq 2/\sqrt{5}$ is likely due to breakdown of ergodicity. In fact, for $\kappa<0$, the billiard is semidispersing, hence ergodic. For $\kappa>2/\sqrt{5}$, Bunimovich's   condition  (see \cite{buni}) for ergodicity based on defocusing in nowhere dispersing billiards holds. This observation offers a clue  to help  identify  lensed billiard systems that fail to be ergodic in   parametric families, although we do not pursue the determination of ergodicity (or failure of ergodicity) in   detail in the present paper except to note it in  some obvious cases.

Let us now  investigate the effect of setting $C<1$. Figure \ref{Exper1_graph2} refers to the same family as in  Figure \ref{shape1} (or  Figure \ref{Exper1_graph1}) except that the fixed parameters are
$\kappa=-1$, $\ell_1=1$, $\ell_2=2$ on the left and  $\kappa=1$,  $\ell_1=1$, $\ell_2=4$ on the right. The potential   $C$ is the  parameter to be varied.
Comparison with the previous graph suggests a discontinuity at $C=1$ (which corresponds to a standard reflecting billiard system.) 
Sample sizes and orbit length for each value of $C$ are as in the previous experiment.  For $0<C<1$ and $\kappa=-1$, a fraction of arrivals from the right at the line of potential function discontinuity   undergoes dispersing reflections. This and the sharper appearance of the graph in that range may  indicate that the system is ergodic and that for $C\leq 0$ it  is not.

\begin{figure}[htbp]
\begin{center}
\includegraphics[width=5in]{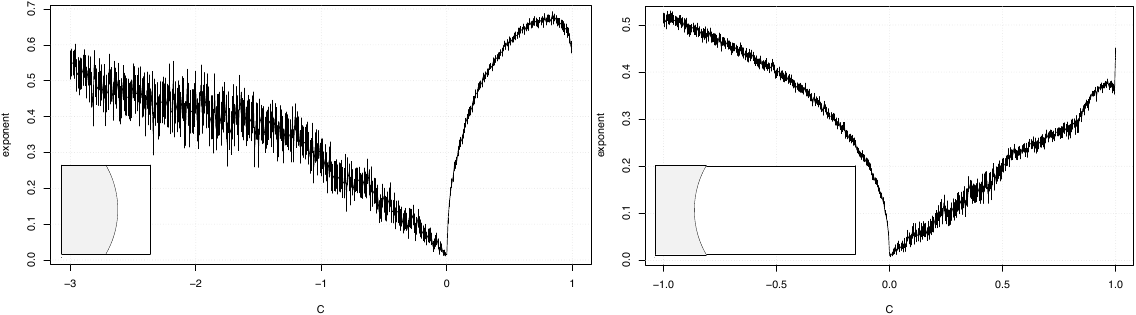}\ \ 
\caption{{\small Family of billiard tables of Figure \ref{shape1}
with    $C$ over the interval $[-3, 1)$.
On the left, $\kappa=-1$, $\ell_1=1$, $\ell_2=2$; on the right, $\kappa=1$, $\ell_1=1, \ell_2=4$.
 The ``fat tail'' on the left for $C\leq 0$ and high variance on the right for part of the range $0\leq C< 1$ suggest that the billiard system is not ergodic in those ranges  of $C$.  
 The larger $\ell_2$ on the right was chosen under the expectation that, for $C$ sufficiently close to $1$,  the defocusing mechanism causing ergodicity will come into play. The somewhat  more clearly defined shape  of the graph on the right roughly in the range $0.8<C<1$ seems to justify this expectation.}}
\label{Exper1_graph2}
\end{center}
\end{figure}

Referring now to the system associated to the graph on the right of Figure \ref{Exper1_graph2},  for negative $C$, the billiard particle can become momentarily trapped
 in the shaded region of Figure \ref{shape1}, where it behaves as in a semidispersing billiard. This suggests that for $\kappa >0$ and $C<0$ this lensed billiard 
may be ergodic. We see again, by comparison with Figure \ref{Exper1_graph1}, a discontinuity at $C=1$. An explanation for the jump discontinuity will be provided shortly.  

Another feature that can be noted on both graphs and others shown later is the presence of a local maximum for positive values of $C$ less than but  close to the point of exponent  discontinuity ($C=1$). We will have more to say about this shortly.

For $C<0$ the exponent grows with
$\sqrt{|C|}$. We leave a detailed proof of this and other remarks to a future paper. However,  the following  elementary observation can be adduced to justify this claim.
Let $N(\xi)$ denote the number of collisions with the boundary of   the shaded  region $\mathcal{A}_0$ (Figure \ref{Exper1_graph2} on the right) that a trajectory with initial state $\xi$ 
undergoes during a sojourn in $\mathcal{A}_0$. Let $T(\xi)$ denote the corresponding time of that sojourn. Now consider
$$\frac1{T(\xi)}\log \left\|d\mathcal{T}^{N(\xi)}w\right\|=\frac{N(\xi)}{T(\xi)}\frac1{N(\xi)}\log \left\|d\mathcal{T}^{N(\xi)}w\right\|. $$
As $C\rightarrow -\infty$,  both $N(\xi)$ and $T(\xi)$ grow to infinity for almost all $\xi$.  It can be expected that the limit on the left will be the exponent for the billiard flow while
the limit on the right will be the product of the exponent for the billiard map times the limit of $N(\xi)/T(\xi)$. 
Under the assumption that the system is indeed ergodic for negative $C$, this quotient converges to $\langle N\rangle/\langle T\rangle$ which, 
by Proposition \ref{averages_2}, is $ \mathcal{s} L/\pi A$, where $L$ is the perimeter of $\mathcal{A}_0$, $A$ is its area, and $\mathcal{s}=\sqrt{(E-C_0)/(E-C_1)}=\sqrt{1+|C|}$.   

\begin{figure}[htbp]
\begin{center}
\includegraphics[width=2.5in]{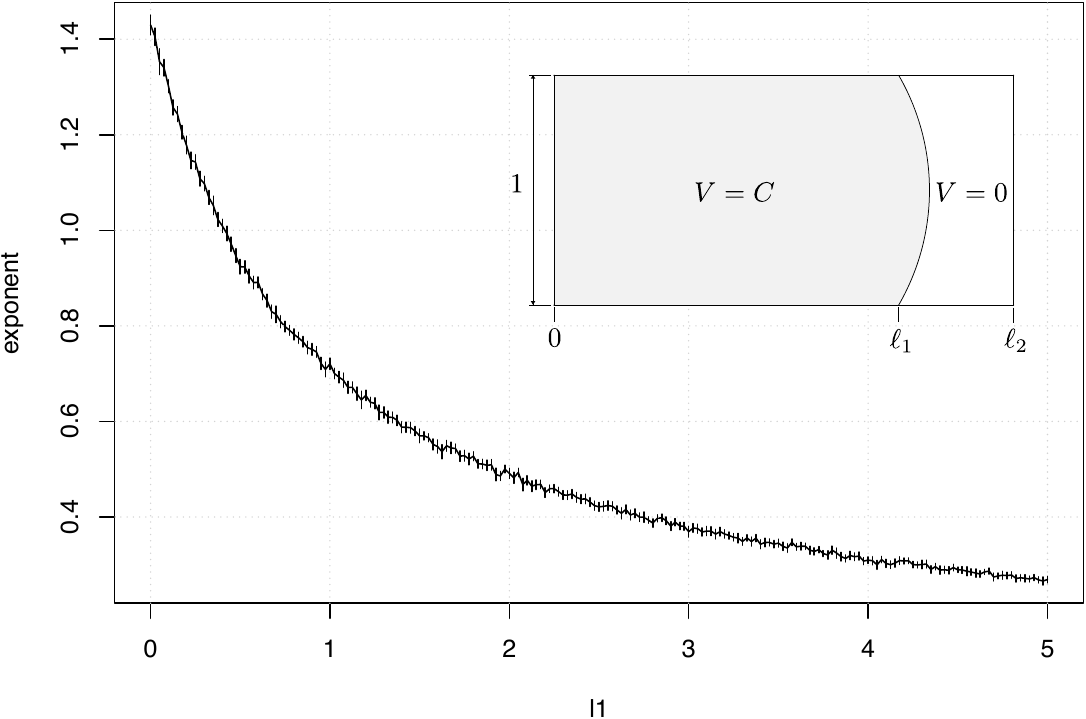}\ \ 
\caption{{\small Here the parameters are: $\kappa=-1$, $C=0.5$, $\ell_2=\ell_1+0.5$, and $\ell_1$ ranges from $0$ to $5$. }}
\label{Exper1_graph4}
\end{center}
\end{figure} 
 
 The dependence on the geometric parameter $\ell_1$ is particularly simple. 
 In Figure \ref{Exper1_graph4} we set $\kappa=-1$, $C=0.5$, $\ell_2=\ell_1+0.5$, and let $\ell_1$ range from $0$ to $5$. A fairly regular dependence of the  exponent on the varying parameter is now observed.
 This  is easily explained by noting  that the number of collisions contributing to the exponent is independent of $\ell_1$, while the time spent in the region where $V=C$ ($0<C<1$) is proportional to $\ell_1$.

It was  observed in the examples that the exponent has a jump discontinuity, as a function of $C$, at the point of transition from lensed to purely reflecting billiard.
Let us, informally, estimate the size of the jump.
We only provide here a heuristic argument, based on Proposition \ref{averages_2} and Corollary \ref{sojourn stat}. (See the example of Figure \ref{BuniExpGraph} for a bit of numerical evidence.)

Suppose the standard billiard systems in the closures of  ${\mathcal{A}}_0$ and ${\mathcal{A}}_1$ are ergodic.
Let  $\chi_\epsilon$ be 
the Lyapunov  exponent of the lensed billiard flow on $\overline{\mathcal{A}}_0\cup \overline{\mathcal{A}}_1$ for which $C_1=0$ and $C_0=E-\epsilon$.
Let $\chi_0$ denote the positive  exponent when $\epsilon =0$. That is, the exponent for the standard  billiard flow in $\overline{\mathcal{A}}_1$, where $C_1=0$. 
Then,  as $\epsilon$ approaches $0$, we expect   $$\chi_\epsilon\rightarrow \frac{A_1}{A_0+A_1}\chi_0$$
to hold,  
where $A_i$ is the area of $\mathcal{A}_i$.  In the example of Figure \ref{BuniExpGraph}, the two areas are equal, therefore as $C\rightarrow 1$ the limit exponent is expected to be
half the value it assumes for $C=1$. This is very nearly what the numerical example shows. (See the right-hand side of that figure.)

To explain this feature, let us first introduce some notation.  Let $\mathcal{R}_i$ be the return map to $\mathcal{C}$ (the intersection of the boundaries of $\mathcal{A}_0$ and $\mathcal{A}_1$) after moving into $\mathcal{A}_i$. Let $\mathcal{N}$ be the number of returns to $\mathcal{C}$ before a trajectory that enters $\mathcal{A}_1$ finally switches 
to $\mathcal{A}_0$. Thus $\mathcal{R}_0$ corresponds to a sojourn into $\mathcal{A}_0$ whereas a sojourn into $\mathcal{A}_1$ involves $\mathcal{N}$ applications 
of $\mathcal{R}_1$. For each positive integer $k$, let
$$\mathcal{F}_k = \mathcal{R}_0\circ \mathcal{R}_1^{\mathcal{N}_k}\circ \cdots\circ \mathcal{R}_0\circ\mathcal{R}_1^{\mathcal{N}_1}. $$
Let us fix a typical orbit, beginning with a $z=(x,v)$ entering $\mathcal{A}_1$:
$$z=z^{(1)}_1,  z^{(0)}_1 := \mathcal{R}_1^{\mathcal{N}_1}(z^{(1)}_1),  z^{(1)}_2:=\mathcal{R}_0(z^{(0)}_1),  \dots,   z_k^{(0)}:=\mathcal{R}_1^{\mathcal{N}_k} (z_k^{(1)}),  
z^{(1)}_{k+1}:=\mathcal{R}_0(z_k^{(0)}), \dots.$$
Note that $\mathcal{N}_j=\mathcal{N}_j(z_j^{(1)})$.  
Times and number of collisions in each sojourn will be written $T^{(i)}$ and $N^{(i)}$ for $i=0, 1$. The total time after $k$ sojourns in $\mathcal{A}_0$ and $k$ sojourns in
$\mathcal{A}_1$ is
$$T_k(z) = T^{(1)}(z^{(1)}_1) + T^{(0)}(z^{(0)}_1)+ \cdots +   T^{(1)}(z^{(1)}_k) + T^{(0)}(z^{(0)}_k)
= T^{(0)}_k(z)  + T^{(1)}_k(z),
$$
where we have defined 
$$T^{(0)}_k(z)=\sum_{j=1}^k T^{(0)}(z^{(0)}_j), \ \ \  T^{(1)}_k(z)= \sum_{j=1}^k T^{(1)}(z^{(1)}_j).$$
We similarly define $N^{(i)}_k(z)$ for the number of collisions after  $k$ sojourns in $\mathcal{A}_0$ and $\mathcal{A}_1$.
Let $\xi$ be a vector defining an infinitesimal variation of the initial condition $z$. We write $\xi^{(1)}_1=\xi$ and, inductively,
$$\xi^{(0)}_j:=\frac{d\mathcal{R}_1^{\mathcal{N}_j}\xi^{(1)}_j}{\left\|d\mathcal{R}_1^{\mathcal{N}_j}\xi^{(1)}_j \right\|}, \ \ \ \xi^{(1)}_j:=\frac{d\mathcal{R}_0^{\mathcal{N}_{j-1}}\xi^{(0)}_{j-1}}{\left\|d\mathcal{R}_0^{\mathcal{N}_{j-1}}\xi^{(0)}_{j-1} \right\|}.$$
Then, omitting reference to $z^{(i)}_j$ to simplify the notation,
$$\frac{\log \|d\mathcal{F}_k\|\xi}{T_k}  = \frac{T_k^{(0)}}{T_k}\frac{N_k^{(0)}}{T_k^{(0)}} \left(\frac1{N^{(0)}_k}\sum_{j=1}^k \log\left\|d\mathcal{R}_0 \xi_j^{(0)}\right\|\right)+ 
\frac{T_k^{(1)}}{T_k}\left(\frac{1}{T_k^{(1)}} \sum_{j=1}^k \log\left\|d\mathcal{R}_1^{\mathcal{N}_j} \xi_j^{(1)}\right\|\right).$$
Focusing attention first on the first summand on the right side of the above equality, the following is expected to hold: the quantity in parentheses is almost everywhere finite, 
the quotient $ {T_k^{(0)}}/{T_k}$ converges to an area ratio based on Proposition \ref{averages_2} and Corollary \ref{sojourn stat} and, due to the same proposition, $N^{(0)}_k/T^{(0)}_k$ should go to $0$. The second
term contains the quotient $T_k^{(1)}/T_k$ which, by Proposition \ref{averages_2} and Corollary \ref{sojourn stat} should limit to $A_1/(A_0+A_1)$. Finally, 
the quantity in parentheses should converge to the exponent $\chi_0$.

Summarizing common features of lensed billiards Lyapunov exponents observed so far:
\begin{itemize}
\item There exists a discontinuity of the Lyapunov exponent  as $C$ approaches the total energy $E$, that is, at the transition from lensed to standard (purely reflecting) billiard. 
\item  In the examples considered so far, the exponent grows as $\sqrt{|C|}$. Different asymptotic behavior  will be noted 
in the additional  examples of Section \ref{BuSiBi}.
\item For values of $C$ less than but close to $E$, the exponent has a local maximum as a function of $C$. This  common feature  is more subtle and its explanation requires a more detailed analysis. See further remarks in Section \ref{lmaxima}.
\end{itemize}

\subsection{Variations on Bunimovich and Sinai billiards}\label{BuSiBi}

Let us consider here a few more examples from different parametric families in which $C$ is the sole parameter, beginning with the variant of the   Bunimovich (half-) stadium shown on the left of Figure \ref{Exper2_graph1}. The positive Lyapunov exponent  is given as a function of the potential function  parameter  $C$.  The standard  stadium billiard corresponds to $C=0$.

We notice greater sample variance for $0<C\leq 1$. In this range, the billiard system is not ergodic. In fact, when the initial position lies in the unshaded square region where the potential is $0$, and the initial velocity is such that the first collision with the vertical line of potential discontinuity is greater than the critical angle, the trajectory must remain in that region since the normal component of the velocity is not large enough to overcome the potential barrier.

\begin{figure}[htbp]
\begin{center}
\includegraphics[width=5.0in]{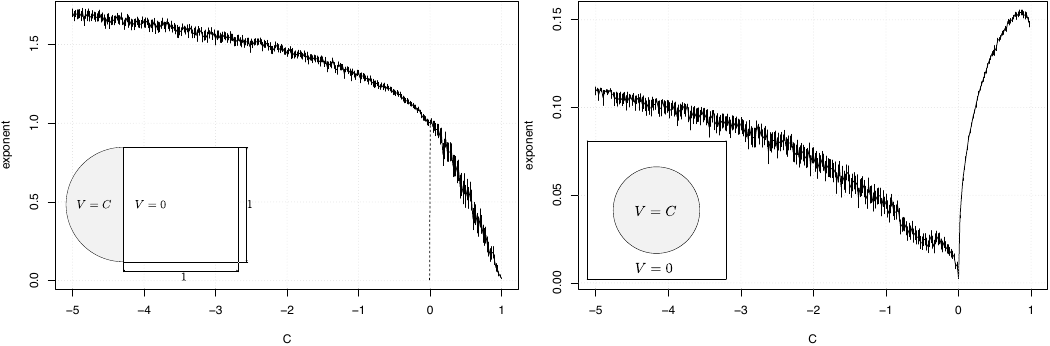}\ \ 
\caption{{\small Lensed Bunimovich and Sinai billiards. The system on the left is not ergodic for $C>0$, and the one on the right is not ergodic for $C\leq 0.$
}}
\label{Exper2_graph1}
\end{center}
\end{figure} 

On the right of Figure \ref{Exper2_graph1} we have a lensed version of Sinai's semidispersing billiard in a square with reflecting boundary. (One should compare this graph with
that on the left of Figure \ref{Exper1_graph2}.)  For all $C\leq 0$ the lensed Sinai billiard is not ergodic since there is a positive measure of trajectories that are trapped inside the circular lens. The condition for being trapped is that the angle which the  initial velocity makes with the inner normal vector to the circle be less than the critical angle.

It will be explained shortly that, in  the limit as $C\rightarrow-\infty$, the Lyapunov exponent for the system on the right for trajectories started in the region of $0$ potential is expected to
converge to the exponent of the Sinai billiard, corresponding to  $C=1$.  We haven't determined  the asymptotic behavior of the graph on the left  for $C\rightarrow -\infty$.

Figure \ref{BuniExpGraph} refers to the lensed half-stadium billiard  shown in the inset of  the graph on the left.
It shows a billiard trajectory for strongly negative $C$.  A billiard trajectory originating from the right side of the stadium must fall into the shaded region at the moment of first arrival at the separating vertical line, and the velocity after crossing must be nearly perpendicular to the line.
Note the  growth of the exponent for negative $C$ for large values of $|C|$.
In these experiments the particle was initiated on the right side of the  stadium with speed $\sqrt{2}$. 
When it moves to the left side, it acquires a large speed, so the exponential rate of growth  of tangent vectors increases accordingly; 
the exponent should be growing proportionally to $\sqrt{|C|}$ as $C\rightarrow -\infty$. On the right: zooming in near $C=1$ shows more clearly the discontinuity at $C=1$.
As expected (given that the two sides of the billiard table have equal area), the jump discontinuity amounts to a near doubling of the exponent.

\begin{figure}[htbp]
\begin{center}
\includegraphics[width=5.0in]{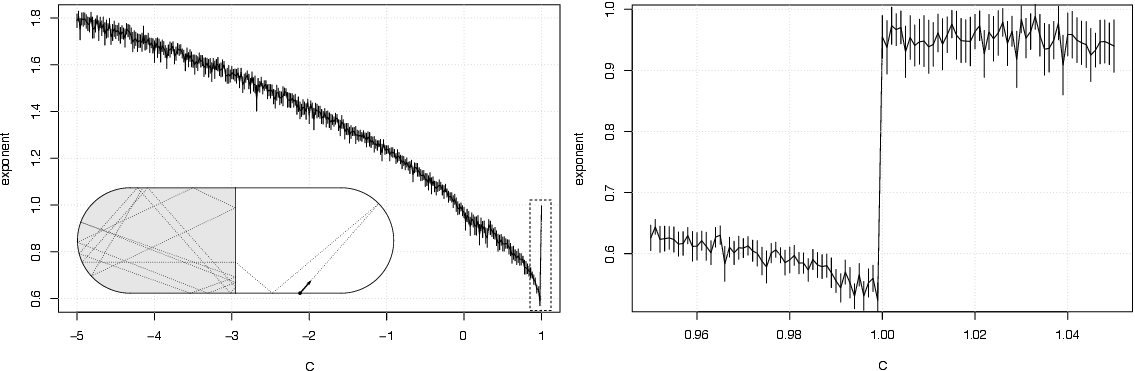}\ \ 
\caption{{\small  Half-lensed Bunimovich stadium} }
\label{BuniExpGraph}
\end{center}
\end{figure}

Let us turn now to the family shown in Figure \ref{BuSiGraph}. This lensed system may be viewed as an interpolation between a focusing billiard, when $C=0$, and a  semidispersing  billiard, when $C=1$.
 One observes  several regimes of behavior over
different ranges of the potential parameter $C$: one regime between $0$ and $1$ (with a discontinuity at $C=1$), one between roughly $-3.3$ and $0$,
and one for $C$ less than approximately $-3.3$.  On this last range, the exponent grows very slowly as $|C|$ increases and appears (for very large values of $|C|$ far outside the range of the graph) to stabilize near   $0.35$, which is roughly the value of the exponent for $C=1$. This apparent coincidence will be explained in Subsection \ref{deep}.

\begin{figure}[htbp]
\begin{center}
\includegraphics[width=3.5in]{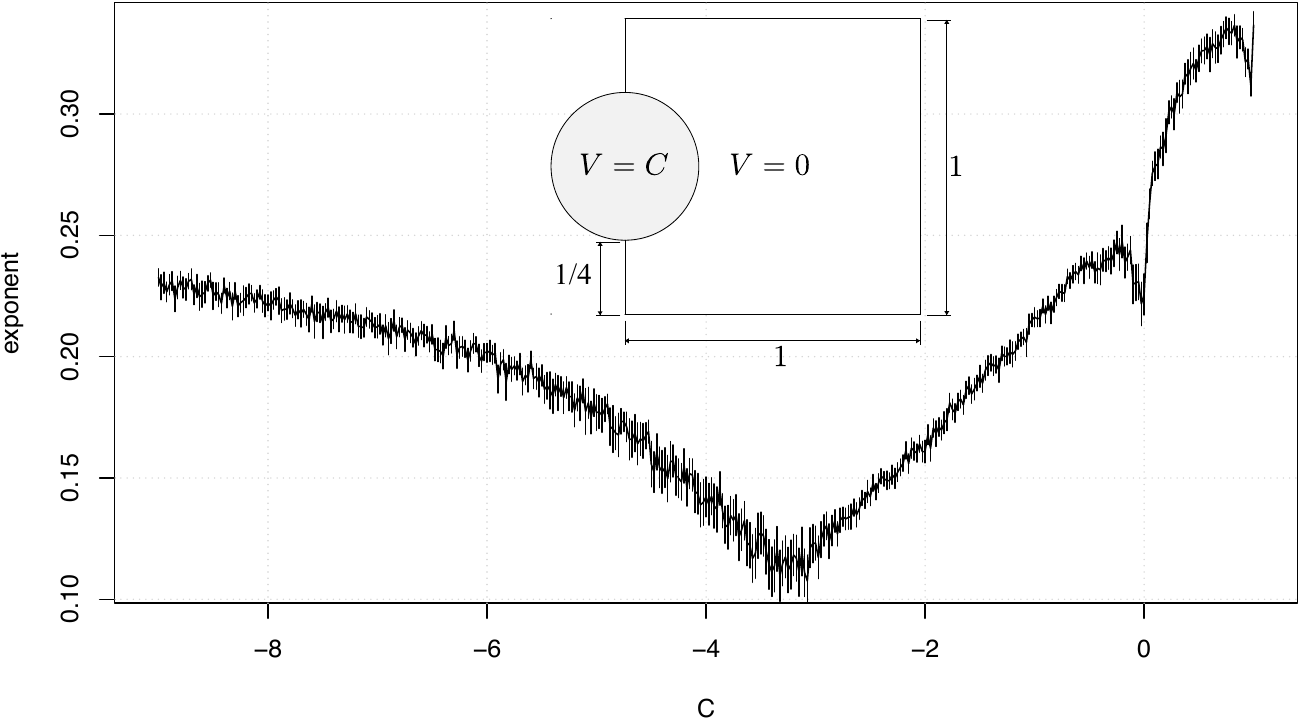}\ \ 
\caption{{\small 
 A lensed interpolation between Bunimovich and Sinai.}}
\label{BuSiGraph}
\end{center}
\end{figure} 

\subsection{Local maxima}\label{lmaxima}
Let us now return to the presence of local maxima for the  Lyapunov exponent as a function of the parameter $C$. This ubiquitous feature is seen  especially clearly in 
the graph on the right-hand side of Figure \ref{Exper2_graph1}, over the interval $0\leq C<E$, for the Sinai lensed billiard system.    The following  comments apply to this system. Though   short of constituting a proof, they contain key ingredients needed for a detailed analysis under more general conditions. We leave this analysis  to a future study. 

When $C=0$, the circular lens has no effect on the motion of the billiard particle and the Lyapunov exponent is $0$. It seems clear that, as $C$ increases from $0$, exponential separation of trajectories should follow due to both (dispersing) reflections and refractions. What is needed then is an explanation for  the observed decrease of the Lyapunov exponent as $C$ increases to $E$ for small values of  $E-C$. 

  \begin{figure}[htbp]
\begin{center}
\includegraphics[width=4in]{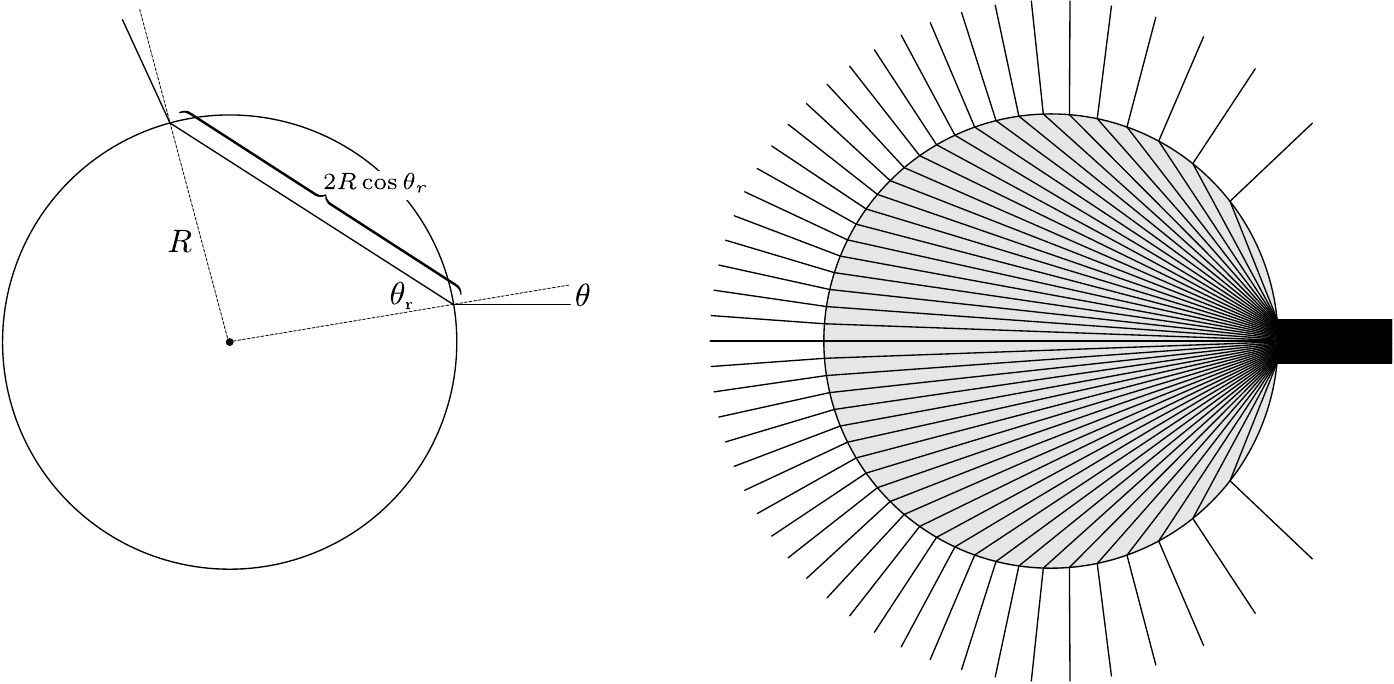}\ \ 
\caption{{\small Left: The time spent by a segment of trajectory inside the lens is $T(\theta)=2R\cos\theta_r/\mathcal{s}$  where $\theta_r$ and
$\theta$ are related by Snell's law and $\mathcal{s}=\sqrt{2(E-C)/m}$ is particle speed. Right: Large expansion factor when $E-C$ is small. The circular arc  consisting of the intersection of the boundary of the  lens and the incident parallel beam  has angle $2\theta_{\text{\tiny crit}}$.  As $C$ approaches $E$, this angle approaches $0$ and
the  directions of trajectories as they exit the lens span an interval of angles approaching $(0,2\pi)$.}}
\label{lens_sinai}
\end{center}
\end{figure} 

We suppose, for the sake of arriving at a rough estimation, that particle collisions with the  lens are statistically independent of each other and the angle of incidence is random and satisfies the (cosine) distribution $\frac12 \cos\theta\, d\theta$. An elementary calculation yields 
probability $p_{\text{\tiny refract}}=\sqrt{\frac{E-C}{E}}$ for a
  collision to result in  refraction, where $E$ is the total energy. The expected expansion rate (of separation of nearby trajectories) in the event of a refraction turns out to be  proportional to $\lambda_\text{\tiny refract}=\sqrt{\frac{E}{E-C}}$. 
 This is the result of an elementary but long calculation, which we omit.  However, this large value,  when $C$ is close to $E$, is easily understood since the particle will undergo a refraction if the angle of incidence is very
small, specifically $|\sin \theta|<\sin \theta_{\text{\tiny crit}} = \sqrt{\frac{E-C}{E}}$, but it fans out inside the lens over the full range $|\theta|<\pi/2$ with the cosine law distribution. (See Proposition \ref{crosscos}.) Figure \ref{lens_sinai}, on the right, illustrates the situation.

The contribution $\tau_{\text{\tiny in}}$
 of the 
motion inside the lens
 to the total time elapsed over a large number of collision events (reflections and refractions) does not depend on $C$. In fact, on one hand, the average time  spent inside the lens during one refraction event is easily shown to be $\frac{\pi R}{\sqrt{2(E-C)/m}}$. This is obtained by noting that 
 the time $T(\theta)$  the segment of trajectory that enters the lens with angle $\theta$ spends inside is $$T(\theta)=2R\cos\theta_r/\sqrt{2(E-C)/m}$$ where $\theta_r$ is
 the refracted angle obtained from $\theta$ by Snell's law (see Figure \ref{lens_sinai}); by Proposition  \ref{crosscos} once again, $\theta_r$ has the cosine distribution. The mean value is then an easy integral calculation. 
On the other hand,  the proportion of refracting collisions is $\sqrt{\frac{E-C}{E}}$, so the term $\sqrt{E-C}$ cancels out, giving the value $\tau_{\text{\tiny in}}=\pi R/\sqrt{2E/m}$. Thus  refractions contribute over many  collisions an expansion rate per collision proportional to
\begin{equation*}\label{refraction term}(\lambda_\text{\tiny refract})^{p_{\text{\tiny refract}}} =\exp\left({\sqrt{\frac{E-C}{E}} \log \sqrt{\frac{E}{E-C}}}\right),\end{equation*}
whose logarithm  decreases to $0$ as $C$ approaches $E$ from below. To this should be added a term that contains the contribution to the Lyapunov exponent due to reflections, which is obtained by standard geometric calculations.  
The result is an expansion factor which, for small $E-C$, has the form
$$\frac{\sqrt{\frac{E-C}{E}} \log \sqrt{\frac{E}{E-C}} + \zeta}{\tau+\tau_{\text{\tiny in }}}, $$
where $\zeta$ is a positive quantity that does not depend on $C$, $\tau$ is the mean time of motion outside of the lens between two consecutive returns to the lens, and 
$\tau_{\text{\tiny in }}$ is the already defined fraction of time of motion inside the lens.  Its presence  in the denominator
explains the previously noted discontinuity of the Lyapunov exponent at $C=E$. 
We conclude that, for small values of $E-C$ and under the  simplifying independence assumption made above, the Lyapunov exponent is expected to decrease to a limit value as $C$ increases towards $E$.

The mechanism suggested here does not account for the second local maximum seen in the example of Figure \ref{BuSiGraph}, in the range $C<0$, $C$ close to $0$. A different analysis is needed, which we won't carry out here.

\subsection{Deep potential well limit} \label{deep} The  examples of lensed billiards discussed in this paper suggest that exploring the asymptotic properties of the exponent   as $C\rightarrow -\infty$ may be a fruitful direction for further study. One may   ask, in particular, whether there are dynamical systems that realize the limit in some sense.  The term {\em deep well} systems will
be used to refer to lensed billiards with strongly negative $C$ or to the limit systems when they can be meaningfully specified.

Let us first look at circular lenses as
in the examples of Figure \ref{BuSiGraph} and the right of  Figure \ref{Exper2_graph1}.  We refer first to the left of Figure \ref{ExplanBuSi}.
As $C$ tends to $-\infty$, a trajectory entering the circular potential well is deflected by refraction towards the radial direction and returns, after a very short time interval (since the speed is very high inside the disc), to a point on the circle which is very close to that from which it entered the circle.  The angle relative to the normal vector
at which the trajectory reaches the right semicircle in the lens boundary after one reflection with the left semicircle  is the same 
as the angle at which the trajectory enters the disc. Therefore the trajectory leaves the disc at the moment of first return to the left semicircle. 
Furthermore, the velocity with which the trajectory exits the circle is very close to what it would be under specular reflection at that point.
 Effectively,  the system behaves in the limit as that  for which $C=E=1$.  The limit trajectory is shown in the figure as a dashed line. The duration of the sojourn inside the disc (moving along the
 radial direction) is zero. Thus it makes sense to say  that the limit system is   the Sinai-type semi-dispersing standard billiard. 

\begin{figure}[htbp]
\begin{center}
\includegraphics[width=3.5in]{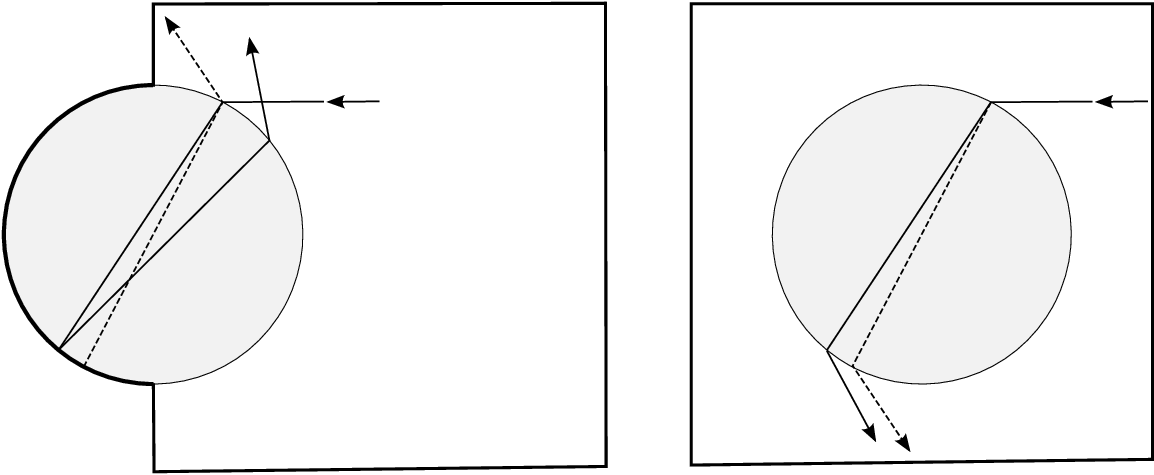}\ \ 
\caption{{\small 
In the limit $C\rightarrow -\infty$,  the system on the left behaves as if $C=1$, corresponding to  a semi-dispersing billiard. For the system on the right, 
the deep well limit is the composition of a standard Sinai billiard (with reflecting walls) and a rotation by $\pi$.}}
\label{ExplanBuSi}
\end{center}
\end{figure}

The same argument applies to the lensed Sinai billiard on the right of Figure \ref{ExplanBuSi}. Here, the deep well limit can be described as follows. Let $\mathcal{T}_0$ be the billiard map of the Sinai billiard in a square with reflecting sides. Let $\mathcal{R}$ be the rotation of the plane $\mathbb{R}^2$ by $\pi$, where the origin is placed at the center of the disc.
Note that these two maps commute and that $\mathcal{R}^2$ is the identity.
Then, disregarding the radial segment  of  trajectory  inside the disc (whose duration is $0$), the limit   system is generated by the map $\mathcal{T}=\mathcal{R}\circ \mathcal{T}_0$.

When the lens region $\mathcal{A}_0$ is such that the standard billiard system in it is ergodic,  the deep well limit may best be described as a random dynamical system of the following type.  We assume for concreteness that the particle mass is $\sqrt{2}$ and the total energy is $E=1$. Let $A$ be the area of $\mathcal{A}_0$
and $a$ the length of the intersection $\mathcal{C}$ of the boundaries of $\mathcal{A}_0$ and $\mathcal{A}_1$. Set $T=\pi A/a$.
Referring  to Figure \ref{DeepWell}, the random system will then be a billiard-like system in $\mathcal{A}_1$  that reflects specularly on the part of the boundary  not including  $\mathcal{C}$ and,  
on reaching  $\mathcal{C}$, trajectories  jump to a  point $x\in \mathcal{C}$  and assume velocity $v$ such that
$(x,v)$ are random variables distributed according to  the Liouville measure: uniform distribution for $x$ and the cosine law distribution for $v$. The jump is not instantaneous but happens
with a random time delay with mean $T$. If the closed billiard system in $\mathcal{A}_0$ is hyperbolic, one may expect this random time to be exponentially distributed. 
This possible description of a deep well billiard limit is motivated by Corollary \ref{sojourn stat} and Proposition \ref{crosscos}, as well as the example of Figure \ref{density_sojourn}.

\begin{figure}[htbp]
\begin{center}
\includegraphics[width=2.0in]{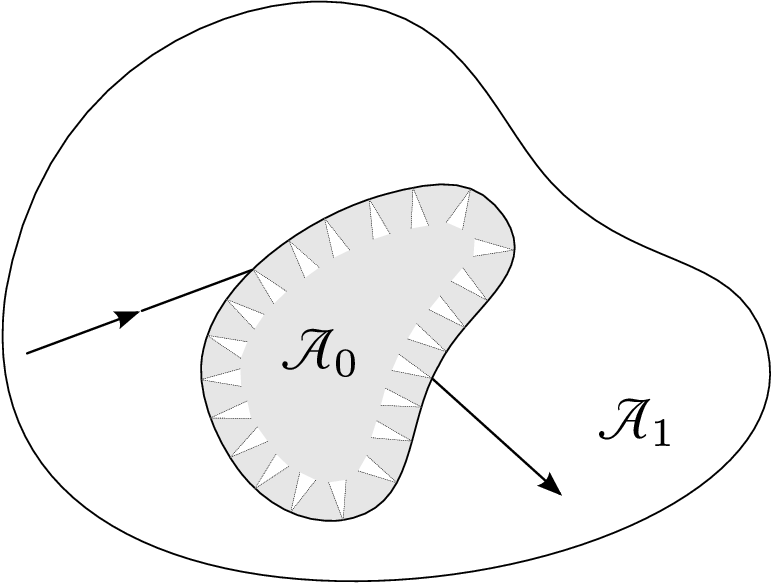}\ \ 
\caption{{\small  A deep well billiard system and the field of thin cones. When the potential in $\mathcal{A}_0$ is strongly negative, a trajectory that falls into $\mathcal{A}_0$  undergoes 
many collisions  with   $\partial \mathcal{A}_0$ during a   time interval having mean value $\pi A/a$
 until the first return to  $\mathcal{C}$  for which the direction of
approach lies  in the field of thin cones defined by $|\sin\theta|< \sqrt{\frac{1}{1+|C|}}$. Here $\theta$ is the angle the velocity of incidence makes with a   normal vector to $\mathcal{C}$ (pointing into $\mathcal{A}_1$) at the collision point.
The trajectory then reemerges into $\mathcal{A}_1$ at a point in $\mathcal{C}$ and velocity   having  
the Liouville measure distribution. 
}}
\label{DeepWell}
\end{center}
\end{figure}

 \appendix
 \section{Appendices}
 \subsection{Motion under discontinuous potential}\label{Ap_A}
In order to justify on physical grounds our definition of the lensed billiard map, 
it will be useful to see how trajectories under discontinuous potentials arise in the  limit of a family of   smooth potential functions with increasingly sharp transition between two constant values. We do this in the setting of Riemannian manifolds  of arbitrary dimension.

The following considerations will be local  in nature. Suppose  that the potential  function $V$, restricted to a neighborhood $\mathcal{U}$ in  the Riemannian manifold $M$ with a smooth metric $\langle\cdot,\cdot\rangle$, has only two values, $C_0$ and $C_1$. The discontinuity of $V$ lies on a smooth hypersurface $\mathcal{S}$, and  $\mathcal{U}\setminus \mathcal{S}$
is the union of open sets $\mathcal{U}_0$ and $\mathcal{U}_1$ such that  $V|_{\mathcal{U}_i}=C_i$.  We make the assumption that  $C_1>C_0$. The discussion in this section applies to the opposite inequality with minor modifications.
Let $\mathbf{n}$ be a unit vector field on $\mathcal{S}$, perpendicular to $\mathcal{S}$, and pointing into $\mathcal{U}_1$. Define for $\epsilon>0$ the set
$$ \mathcal{R}_\epsilon:=\{\exp_{x}(s \mathbf{n}_x): x\in \mathcal{S}\cap \mathcal{U}, s\in [0,\epsilon]\}\cap\overline{\mathcal{U}}_1.$$
Then $\mathcal{R}_\epsilon$ is the union of submanifolds $\mathcal{S}_\epsilon(s)$ consisting of the points $\exp_{x}(s \mathbf{n}_x)$ in $\mathcal{R}_\epsilon$ 
in which $x\in \mathcal{S}$ and  $s$ is constant. Note that $\mathcal{S}_\epsilon(0)=\mathcal{S}$. 
The vector field $\mathbf{n}$ can be extended to all of $\mathcal{R}_\epsilon$
by setting $\mathbf{n}_{\gamma(s)}=\gamma'(s)$ where $\gamma(s)=\exp_x(s\mathbf{n}_x)$ and $x\in \mathcal{S}$. It is not difficult to show (this is essentially Gauss's lemma, \cite{dC})
that $\mathbf{n}_y$ is perpendicular to $\mathcal{S}_\epsilon(s)$ at any $y\in \mathcal{S}_\epsilon(s)$. Clearly $D_{\mathbf{n}}\mathbf{n}=0$ since the integral curves of $\mathbf{n}$ are geodesics. (Here $D$ denotes the Levi-Civita connection.)

We  define a smooth potential function $V_\epsilon$ on $\mathcal{U}$ as follows. Let $f_\epsilon:\mathbb{R}\rightarrow [C_0, C_1]$ be a smooth real-valued increasing function such that 
$$f_\epsilon(s)=\begin{cases}
C_0 & \text{ if } s\leq 0\\
C_1 & \text{ if } s\geq \epsilon.
\end{cases}
$$
 Now set  
\begin{equation}\label{potential_eps}
V_\epsilon(x)=\begin{cases}
 C_0 &\text{ if } x\in \mathcal{U}_0\\
 C_1 &\text{ if } x\in \mathcal{U}_1\setminus\mathcal{R}_\epsilon\\
 f_\epsilon(s) & \text{ if } x \in \mathcal{S}_{\epsilon}(s).
 \end{cases}
\end{equation}

Let $\mathbb{S}_x$ denote the shape operator of the level hypersurfaces $\mathcal{S}_\epsilon(s)$. Thus, by definition,
$$\mathbb{S}_xv=-D_v\mathbf{n}$$
for all $v\in T_x\mathcal{S}_\epsilon(s)$. 
Finally, if $v\in T_x\mathcal{R}_\epsilon$, the orthogonal decomposition of $v$ into
a tangent vector $v_\tau$ to the hypersurface  $\mathcal{S}_\epsilon(s)$ containing $x$ and a perpendicular vector will be written as
$v=v_\tau+v_n\mathbf{n}_x$ 

\begin{lemma}\label{NewtonSplit}
Newton's equation in $\mathcal{R}_\epsilon$ with potential function $V_\epsilon$ as defined in Equation (\ref{potential_eps}) decomposes orthogonally as
\begin{align*}
\frac{Dv_\tau}{dt} &=\langle v_\tau, \mathbb{S}_x v_\tau \rangle \mathbf{n} + v_n \mathbb{S}_x v_\tau\\
-\frac{d{v}_n}{dt} &= \langle v_\tau, \mathbb{S}_x v_\tau \rangle +\frac1m f_\epsilon'(s(x)).
\end{align*}
\end{lemma}
\begin{proof}
Since $D_{\mathbf{n}}\mathbf{n}=0$, we have 
$$\frac{D \mathbf{n}}{dt} = D_{v} \mathbf{n} =  D_{v_\tau} \mathbf{n}  + v_n D_{\mathbf{n}}\mathbf{n} =-\mathbb{S}_xv_\tau.  $$
$$\left\langle \frac{Dv_\tau}{dt},\mathbf{n}\right\rangle = -\langle v_\tau, D_v\mathbf{n}\rangle = \langle v_\tau, \mathbb{S}_x v_\tau\rangle. $$
Let $\Pi_x$ denote the orthogonal projection to the tangent space at $x$ to the hypersurface $\mathcal{S}_\epsilon(s)$ containing $x$. Then
\begin{align*}
-f_\epsilon'(s)\mathbf{n}&=-\text{grad }V\\
&=m\frac{Dv}{dt}\\
&=m\left(\frac{Dv_\tau}{dt} +\dot{v}_n\mathbf{n} - v_n\mathbb{S}_x v_\tau\right)\\
&=m\left(\Pi\frac{Dv_\tau}{dt} +\langle v_\tau, \mathbb{S}_x v_\tau\rangle \mathbf{n} +\dot{v}_n\mathbf{n} - v_n\mathbb{S}_x v_\tau\right).
\end{align*}
Separating the normal and tangential parts we obtain  $\Pi{Dv_\tau}/{dt} =v_n\mathbb{S}_x v_\tau$ and the desired equations.
\end{proof}

\begin{lemma} We make   the same assumptions as in Lemma \ref{NewtonSplit},  except that now we require $V$ to be constant equal to $C_0$ on $\mathcal{U}_0$
and constant equal to $C_1$ on $\mathcal{U}_1\setminus\mathcal{R}_\epsilon$.  
Let the initial velocity of a particle that enters $\mathcal{R}_\epsilon$ from $\mathcal{U}_0$ be $v(0)=v^-$ and let the velocity upon exit from $\mathcal{R}_\epsilon$ be
 $v^+$ for the discontinuous potential function $V$ and $v_\epsilon^+$ for the potential $V_\epsilon$. Further assume that the shape operator of the hypersurface $\mathcal{S}$ is bounded.  Then $v^+=v^+_\epsilon +O(\epsilon)$ and the return time to the boundary of $\mathcal{R}_\epsilon$ is $T_\epsilon=O(\epsilon)$. 
\end{lemma}
\begin{proof}
We consider the case $C_1>C_0$. The opposite inequality can be argued similarly. Thus $V_\epsilon$ is increasing along the radial direction (parallel to the vector field $\mathbf{n}$) in $\mathcal{R}_\epsilon$. Energy conservation implies 
$$\|v(t)\|=\sqrt{2({E}-V_\epsilon(x(t)))/m}\leq \|v^-\|$$
for all $t$. We are only concerned with the trajectory $x(t)$ from $t=0$ to $t=T_\epsilon$, when it reaches again the boundary of $\mathcal{R}_\epsilon$ at a point where $V_\epsilon=C_1$. (The case in which the trajectory does not overcome the potential barrier and returns to the boundary of $\mathcal{U}_0$ can be dealt with by similar arguments.)
Thus we know that 
$$B:=\frac{(v_n^-)^2 -\frac2m(C_1-C_0)}{4 \|v^-\|^3 K}>0 $$
where $K$ is an upper bound on the norm of the shape operator. By Lemma \ref{NewtonSplit},
$$\frac{d}{dt}\left(\frac12 m v_n^2\right)= v_n \left[-m\langle v_\tau, \mathbb{S}_x v_\tau\rangle - f'_\epsilon(s(t))\right].$$
Note that $v_n=\dot{s}$, so $ f'_\epsilon(s(t))v_n\, dt=d(f_\epsilon(s(t))).$ Integrating    in $t$ gives
\begin{align*}
v_n^2(t) &= v_n^2(0) -\frac2m\left[f(s(t))-C_0\right] -2\int_0^t v_n(u) \langle v_\tau(u),\mathbb{S}_{x(u)}v_\tau(u)\rangle\, du\\
&\geq v_n^2(0) -  \frac2m(C_1-C_0)-2\|v^-\|^3 K t\\
&= \frac{\left(v_n^-\right)^2 -\frac2m (C_1-C_0)}{4\|v^-\|^3K} 4\|v^-\|^3 K-2\|v^-\|^3 Kt\\
&=2\|v^-\|^3 K(2B-t). 
\end{align*}
This quantity is bounded away from $0$ for $t\in [0,B]$. 
Explicitly,
$$v_n(t)\geq \sqrt{2\|v^-\|^3 KB} $$
in that interval.
 Then
$$\int_0^B v_n(t)\, dt \geq 2\sqrt{\|v^-\|^3 KB} \int_0^B \sqrt{1-\frac{t}{2B}}\, dt=\frac83\left(1-\frac1{2\sqrt{2}}\right)\sqrt{\|v^-\|^3 K B^3}=:A>0. $$
Setting $\epsilon <A$, we can be certain that the trajectory $x(t)$ will reach the boundary of $\mathcal{R}_\epsilon$ (on the side of $\mathcal{U}_1$)
at $s=\epsilon$ in time $T_\epsilon\leq B$ and that
$$\epsilon = \int_0^{T_\epsilon} v_n(t)\, dt \geq \sqrt{2\|v^-\|^3 K B} T_\epsilon. $$
This shows that  
$$T_\epsilon \leq \frac{\epsilon}{\sqrt{2\|v^-\|^3 K B}}=O(\epsilon). $$
It is now a simple consequence of Lemma \ref{NewtonSplit} that
\begin{equation}\label{eqapp}v_\tau(T_\epsilon) = v_\tau(0)+O(\epsilon), \ \ v_n(T_\epsilon) = \sqrt{v_n(0)-\frac2m(C_1-C_0)} +O(\epsilon). \end{equation}
Therefore $v^+=v^+_\epsilon +O(\epsilon)$ as claimed.
(More properly, one may write the first equation in Lemma \ref{NewtonSplit} as a system of first order non-linear equations in the components of $v_\tau$ with respect to an orthonormal parallel frame of vector fields
along a trajectory $x(t)$. The approximation given above in (\ref{eqapp}) is easily shown to hold for these components.)
\end{proof}

\subsection{Invariant lensed billiard measure}\label{Ap_B}
Let us begin by recalling a basic fact about the invariance of the canonical (Liouville) measure under the billiard map.
Suppose $\mathcal{R}$ is a Riemannian manifold of dimension $n$ with piecewise smooth boundary and $V:\mathcal{R}\rightarrow \mathbb{R}$ is a potential function which  is bounded from above. We assume for the moment that $V$ is smooth and let $\mathcal{T}$ be the billiard map on the phase space $\mathcal{V}$. (When $V$ is smooth,  the set $\mathcal{A}$ plays no role and 
 $\mathcal{R}$ and $\mathcal{B}$ are the same.) We fix a value $E$ for the total energy such that $E-V$ is bounded from below by a positive number and let $\mathcal{V}_E\subseteq \mathcal{V}$ denote the level set of the total energy function for the value $E$. Thus 
$$\mathcal{V}_E:=\left\{(x,v)\in \mathcal{V}:x\in \partial\mathcal{R},  \|v\|=\sqrt{2(E-V(x))/m} \right\}. $$
The {\em canonical billiard measure}, or  {\em Liouville measure}, is the measure on $\mathcal{V}_E$  obtained by the restriction to $\mathcal{V}_E$ of the 
volume form derived from the symplectic form on the tangent bundle of $\mathcal{R}$. Let us denote this measure by $\nu$. It is a classical result that this measure is invariant under the billiard map.

More generally, if $\mathcal{S}$ is a codimension $1$ submanifold in the interior of $\mathcal{R}$, we can redefine $\mathcal{T}$ so that
it is the map
defined over the union of the boundary of $\mathcal{R}$ and $\mathcal{S}$ and, upon reaching $\mathcal{S}$, the velocity of the billiard flow is not reflected. We have in mind to apply this fact to $\mathcal{S}$ that are level sets of $V$. 

\begin{figure}[htbp]
\begin{center}
\includegraphics[width=2in]{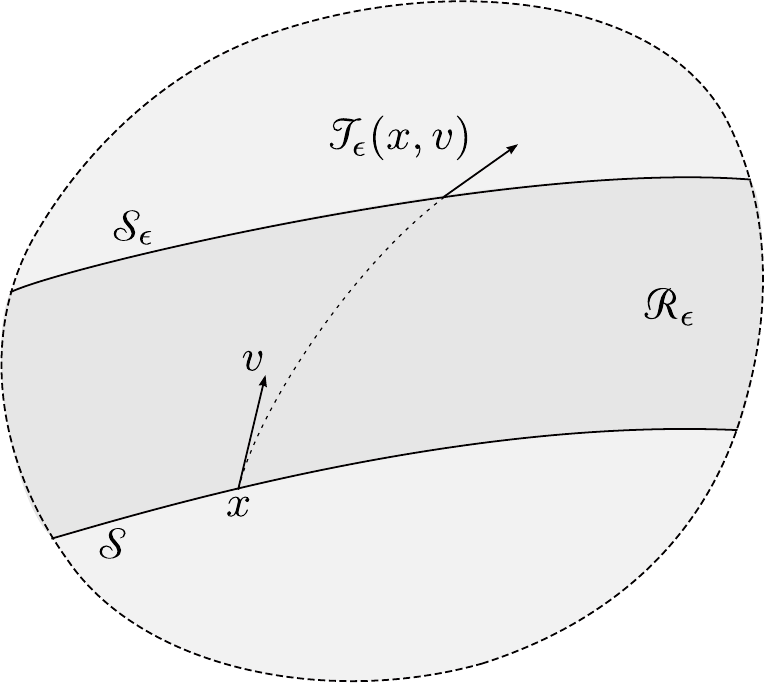}\ \ 
\caption{{\small  Set-up for Theorem \ref{invariance}. The submanifold $\mathcal{S}$ is a hypersurface of discontinuity of the potential function.  When $\epsilon$ approaches $0$, $\mathcal{S}_\epsilon$ limits to a piece of the boundary of    $\mathcal{A}_1$  while $\mathcal{S}$ is identified with the corresponding boundary piece of $\mathcal{A}_0$.}}
\label{m_invariance}
\end{center}
\end{figure}

Let $(x,v)\in \mathcal{V}_E$ (a point in the phase-space of the lensed billiard) be such that $x$ is a regular point  in a smooth boundary piece $\mathcal{S}$ of $\mathcal{A}_0$ and  $\langle \mathbf{n}_x, v\rangle<0$ where $\mathbf{n}$ is the normal vector field on $\mathcal{S}$ pointing to the interior of $\mathcal{A}_0$. Thus $(x,v)$ represents the state of a billiard particle arriving at a point of discontinuity of the potential.  We suppose that the normal component of $v$  is sufficiently large for the billiard flow trajectory to overcome the
potential barrier and cross into $\mathcal{A}_1$. Let $\pi:\mathcal{V}_E\rightarrow \mathcal{B}$ be the base-point projection and $\mathcal{U}$ 
a neighborhood of $(x,v)$  such that $\mathcal{U}\cap\pi^{-1}(\mathcal{S})$  consists  of regular points and velocities leading to refraction. 

Our goal is to show that
the lensed billiard map preserves the canonical measure $\nu$. For this it is enough   to show that the refraction operation $(x,v^-)\mapsto (x,v^+)$ itself preserves the measure.
It will be convenient to denote this operation   by $\mathcal{T}$ even though it does not involve the displacement which, together with the reflection or refraction, makes up the  billiard map as defined in Section \ref{lb}. Let $C_i$ be the value of the potential function on $\mathcal{A}_i$. As in Appendix  \ref{Ap_A},  we approximate the potential step  
by a smooth transition $V_\epsilon$ over a narrow band $\mathcal{R}_\epsilon\subset \mathcal{A}_1$ bounded   by $\pi(\mathcal{U})\cap \mathcal{S}$ and $\pi(\mathcal{U})\cap\mathcal{S}_\epsilon$. This is indicated in Figure \ref{m_invariance}. Thus $V_\epsilon$ has value $C_0$ on $\pi(\mathcal{U})\cap\mathcal{S}$ and $C_1$ on $\pi(\mathcal{U})\cap\mathcal{S}_\epsilon$,
and it is defined by means of a function $f_\epsilon$ according to the construction used in Appendix \ref{Ap_A}. Let   $$\mathcal{T}_\epsilon:\widetilde{\mathcal{S}}:=\mathcal{U}\cap\pi^{-1}(\mathcal{S})\rightarrow \widetilde{\mathcal{S}}_\epsilon:=\mathcal{U}\cap\pi^{-1}(\mathcal{S}_\epsilon)$$ 
be the map induced by the Hamiltonian flow. (See Figure \ref{m_invariance}.)  Note that   $\mathcal{S}_\epsilon$   limits to $\mathcal{S}$ in $\mathcal{R}$ while in $\mathcal{B}$ it
limits to the subset of the boundary of $\mathcal{A}_1$ corresponding to $\mathcal{S}$ as $\epsilon$ approaches $0$; and 
 $\mathcal{S}$ should be viewed as contained in the boundary of $\mathcal{A}_0$. Thus in the limit we distinguish $\widetilde{\mathcal{S}}^-:=\widetilde{\mathcal{S}}$, contained in the domain of $\mathcal{T}$ and $\widetilde{\mathcal{S}}^+:=\lim_{\epsilon\rightarrow 0}\widetilde{\mathcal{S}}_\epsilon$, contained in the range of $\mathcal{T}$.
Finally, let $\nu_\epsilon$ denote the canonical measure on $\widetilde{\mathcal{S}}_\epsilon$. Note that measure invariance under the Hamiltonian flow implies ${\mathcal{T}_\epsilon}_*\nu=\nu_\epsilon$, where ${\mathcal{T}_\epsilon}_*$ is the   push-forward operation on measures. 

With these notations, proving invariance of the canonical measure under the lensed billiard map reduces to showing the following. For any given  continuous real-valued function $\varphi$ with compact support in $\mathcal{U}$, 
$$
\int_{\widetilde{\mathcal{S}}_\epsilon} \varphi \, d\nu_\epsilon =\int_{\widetilde{\mathcal{S}}_\epsilon} \varphi \, d({\mathcal{T}_\epsilon}_*\nu)
=\int_{\widetilde{\mathcal{S}}} \varphi\circ \mathcal{T}_\epsilon\, d\nu\rightarrow \int_{\widetilde{\mathcal{S}}} \varphi\circ \mathcal{T}\, d\nu=\int_{\widetilde{\mathcal{S}}^-} \varphi\circ \mathcal{T}\, d\nu.
$$
as $\epsilon\rightarrow 0$. On the other hand, the left-most  term above converges to 
$$\int_{\widetilde{\mathcal{S}}_\epsilon} \varphi \, d\nu_\epsilon\rightarrow  \int_{\widetilde{\mathcal{S}}^+} \varphi \, d\nu.$$
Therefore
$$\int_{\widetilde{\mathcal{S}}^-} \varphi\circ \mathcal{T}\, d\nu= \int_{\widetilde{\mathcal{S}}^+} \varphi \, d\nu.$$
This implies invariance of the canonical measure under refraction.  We thus arrive at the following.

\begin{theorem}\label{invariance}
The lensed billiard map preserves the Liouville measure.
\end{theorem}

The invariant measure can be given the following concrete form. (See \cite{cook}.) Let $\mu$ denote the Riemannian volume measure on the boundary of $\mathcal{B}$ and 
$\sigma_x$ the Riemannian measure on the  hemisphere   $$\{v\in T_x\mathcal{B}:x \text{ is a regular point in } \partial \mathcal{B}, \|v\|=1, \langle \mathbf{n}_x,v\rangle\geq 0\}.$$
Then, up to multiplicative constant,
$$d\nu(x,v) =\left[2(E-V(x))\right]^{\frac{n-1}{2}} \left\langle \mathbf{n}_x,{v}/{\|v\|}\right\rangle\,  d\sigma_x(v/\|v\|)d\mu(x). $$
In particular, the  distribution of directions has density (relative to the Riemannian volume on the unit hemisphere)  given by the cosine of the angle $\theta$ between
$\mathbf{n}_x$ and $v$.

\begin{proof}[Proof of Proposition \ref{crosscos}] In this general Riemannian setting, Proposition \ref{crosscos} can be proved as follows.
Let $S^+_x$ be the hemisphere centered at $x$ consisting of unit vectors $u$ such that $\langle \mathbf{n}_x,u\rangle\geq 0$, and let $D$ be the unit disc in the
tangent space $T_x\mathcal{C}$.  Let $(e_1, \dots, e_{n-1})$ be an orthonormal basis of $T_x\mathcal{C}$ and let $y=(y_1,\dots,y_{n-1})$ denote coordinates
of points in $D$ relative to this basis.  If we parametrize $S^+_x$ using these coordinates, the volume element $dV$ on the hemisphere satisfies  $ dV(u)= (\cos\varphi)^{-1}dy_1\dots dy_{n-1}$
where $\varphi$ is the angle between $u\in S^+_{x}$ and $\mathbf{n}_x$.  In other words, the cosine law corresponds to  uniform distribution  in $D$. 
The set of  trajectory segments arriving at $x$ from $\mathcal{A}_1$ whose velocities  undergo refraction are those for which $|\sin \theta|<r_0$.  Under the given parametrization 
of $S^+_x$ by $D$, this set defines the open  disc  $D_0\subseteq D$ centered at the origin of $T_x\mathcal{C}$ with radius  $r_0$. Under the 
cosine law,  velocities undergoing refraction correspond to points uniformly distributed in $D_0$.  If $v^+$ is the velocity after refraction, the relation  $\sin\varphi = r_0^{-1}\sin\theta$ implies that the orthogonal projection of $v^+/|v^+|$ to $T_x\mathcal{C}$ has the uniform distribution on $D$,
hence satisfies the cosine law.
\end{proof}

 \subsection{Sojourn mean values}\label{Ap_C}
 We begin by recalling notation used in Section \ref{switching}, with the difference that here our billiard domains are $n$-dimensional.
Let ${\mathcal{A}}$ be either ${\mathcal{A}}_0$ or ${\mathcal{A}}_1$. Let  $V$ and  $A$ be the (Euclidean) $n$-dimensional volume of $\mathcal{A}$ and  $(n-1)$-dimensional volume  of the boundary of $\mathcal{A}$, and $a$ the $(n-1)$-dimensional  volume of the crossing boundary $\mathcal{C}:=\overline{\mathcal{A}}_0\cap \overline{\mathcal{A}}_1$.  
We denote by $\mathcal{V}$   the   space of pairs $(x,v)$ where $x\in \partial \overline{\mathcal{A}}$ and $v$ is a (velocity) tangent vector to $\overline{\mathcal{A}}$ at $x$ pointing into $\mathcal{A}$ having norm (speed) $\mathcal{s}.$
Let  $\mathcal{E}$ be  the space of pairs $(x,v)$ where now $x\in \mathcal{C}$.  On $\mathcal{E}$ we define the first return billiard map, $R(x,v)$. It is well-defined on a subset of full (Liouville) measure
 due to Poincar\'e recurrence and the assumption that $\mathcal{A}$ is bounded. 
 
 At each $(x,v)\in \mathcal{E}$, let  $T(x,v)$ and $N(x,v)$  denote, respectively,  the time of first return to $\mathcal{E}$ and the number of collisions  with the boundary of $\overline{\mathcal{A}}$
  of a billiard trajectory with initial state $(x,v)$
  before returning $\mathcal{E}$. For each $(x,v)\in \mathcal{V}$, let $\tau(x,v)$ denote the time duration of free flight from $(x,v)$  to the point of next collision. This is
 naturally the length of the free flight divided by the speed $\mathcal{s}$. Finally, denoting by $\nu$ and $\nu_\mathcal{E}$ the normalized Liouville measure on $\mathcal{V}$ and
 $\mathcal{E}$, respectively, we introduce the mean values
$$
 \langle N\rangle_{\mathcal{E}} :=\int_{\mathcal{E}} N(x,v)\, d\nu(x,v), \ \  
 \langle T\rangle_{\mathcal{E}}:=\int_{\mathcal{E}} T(x,v)\, d\nu(x,v), \ \  
 \langle \tau\rangle_{\mathcal{V}} :=\int_{\mathcal{V}} \tau(x,v)\, d\nu(x,v).
$$

\begin{figure}[htbp]
\begin{center}
\includegraphics[width=2.5in]{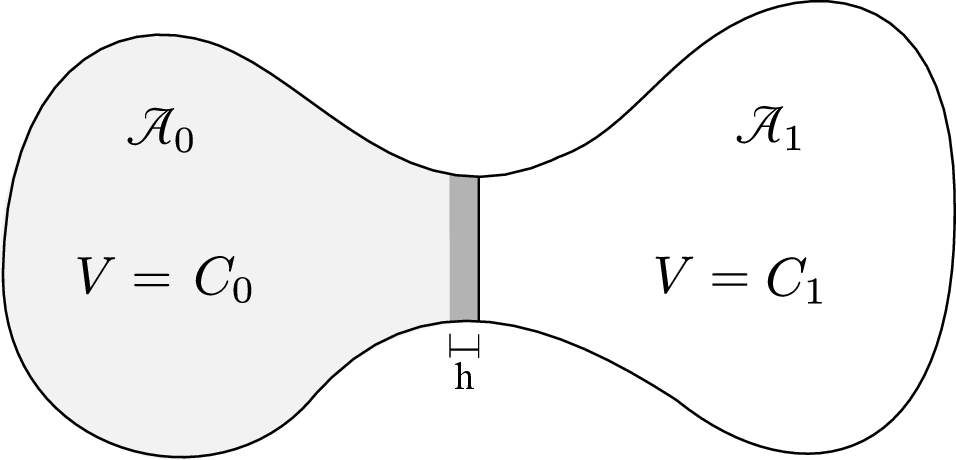}\ \ 
\caption{{\small Set-up for the switching dynamics description of lensed billiards. The collar region of width $h$ around $\overline{\mathcal{A}}_0\cap\overline{\mathcal{A}}_1$ contained in $\mathcal{A}_0$, is  used in the proof of Theorem \ref{averages}.
 }}
\label{two_chambers}
\end{center}
\end{figure}

\begin{theorem}\label{averages} With the notations just introduced and 
under the assumption that the standard billiard map in $\overline{\mathcal{A}}$ is ergodic, the following relations hold:
\begin{enumerate}
\item $\langle N\rangle_{\mathcal{E}} = \frac{A}a$;
\item $\langle T\rangle_{\mathcal{E}}= \sqrt{\pi} n\frac{\Gamma\left(\frac{n}{2}+\frac12\right)}{\Gamma\left(\frac{n}{2}+1\right)} \frac{V}{a\mathcal{s}}$;
\item $\langle \tau\rangle_{\mathcal{V}}= \sqrt{\pi} n\frac{\Gamma\left(\frac{n}{2}+\frac12\right)}{\Gamma\left(\frac{n}{2}+1\right)} \frac{V}{A\mathcal{s}};$
\item $\langle T\rangle_{\mathcal{E}}=\langle N\rangle_{\mathcal{E}}\langle \tau\rangle_{\mathcal{V}}.$
\end{enumerate}
Furthermore, if $\mathcal{E}_0\subseteq \mathcal{E}$ consists of pairs $(x,v)$ such that $x\in \mathcal{C}$ and the angle $\theta$ which $v$ makes with the normal $\mathbf{n}_x$ satisfies
$|\sin\theta|<r_0$ then the mean  number of returns to $\mathcal{E}$ before a first return to $\mathcal{E}_0$ is
$1/r^{n-1}_0$.
\end{theorem}
\begin{proof}
The following is a standard application of the ergodic theorem, which we nevertheless  present in detail. Recall that $\mathcal{T}$ denotes the billiard map on ${\mathcal{V}}$ and $R$ the first return map to $\mathcal{E}$.  Let us write $\xi=(x,v)$. Then
$$T(\xi) =\tau(\xi)+\tau(\mathcal{T}(\xi))+\cdots+\tau(\mathcal{T}^{N(\xi)-1}(\xi)). $$
For each $\xi\in\mathcal{E}$ and positive integer $\ell$, define
\begin{align*}
N^\ell(\xi)&:=N(\xi)+N(R(\xi))+\cdots+ N(R^\ell(\xi));\\
T^\ell(\xi)&:=T(\xi)+T(R(\xi))+\cdots+ T(R^\ell(\xi)).
\end{align*}
Then $N^\ell(\xi)$ is the total number of collisions with the boundary during the period of $\ell$ returns to the distinguished boundary part, and $T^\ell(\xi)$ is the total
time  elapsed during the same period. By the ergodic theorem, we have
$$\lim_{\ell\rightarrow \infty} N^\ell(\xi)/\ell =\langle N\rangle_{\mathcal{E}}, \ \ \lim_{\ell\rightarrow \infty}T^\ell(\xi) = \langle T\rangle_{\mathcal{E}}. $$
Therefore, for all $\xi\in {\mathcal{E}}$ but for a set of zero probability,
\begin{align*}
\langle N\rangle_{\mathcal{E}}^{-1}&=\lim_{\ell\rightarrow\infty}\frac{\ell}{N^\ell(\xi)}\\
&= \lim_{\ell\rightarrow \infty} \frac1{N^\ell(\xi)}\sum_{i=0}^{N^\ell(\xi)} \mathbbm{1}_\mathcal{E}(\mathcal{T}^i(\xi))\\
&=\text{Prob}({\mathcal{E}})\\
&=\frac{a}{A}.
\end{align*}
This shows the first identity. To obtain the fourth, start from
$$ \sum_{k=0}^{N^\ell(\xi)} \tau(\mathcal{T}^k(\xi))=T(\xi)+T(R(\xi))+\cdots+T(R^{\ell-1}(\xi))$$
and average both sides over ${\mathcal{E}}$, using $R$-invariance of the probability measure on ${\mathcal{E}}$ induced by the Liouville measure. This gives
$$\left\langle \sum_{k=0}^{N^\ell(\xi)}\tau(\mathcal{T}^k(\xi))\right\rangle_{\mathcal{E}} =\ell \langle T\rangle_{\mathcal{E}}.$$
Consequently,
$$\langle T\rangle_{\mathcal{E}} = \lim_{\ell\rightarrow \infty} \left\langle \left(\frac{N^\ell(\xi)}{\ell}\right)\left(\frac1{N^\ell(\xi)}\sum_{k=0}^{N^\ell(\xi)}\tau(\mathcal{T}^k(\xi))\right) \right\rangle_{\mathcal{E}}=\langle N\rangle_{\mathcal{E}}\langle \tau\rangle_{\mathcal{V}}.$$
To establish the second relation, let us first introduce the collar region $U_h$ of width $h$  separating $\mathcal{C}$ from the rest of $\mathcal{A}$, where $h$ is a small positive number.  This is shown schematically in Figure \ref{two_chambers}.
Except for a set of small measure, which goes to zero with $h$, the time it takes  for the trajectory with initial condition $\xi=(x,v)\in {\mathcal{E}}$ to traverse $U_h$ is
$\eta(\xi)=h/\langle \mathbf{n}_x, v\rangle$ plus terms of higher order in $h$ due to the possibly non-zero curvature of $\mathcal{C}$ at $x$. An explicit integral calculation  gives
$$ \lim_{h\rightarrow 0}\frac1h \langle \eta\rangle_{\mathcal{E}} =\frac{1}{\mathcal{s}}
\frac{\text{Vol}\left(S_x^+\right)}{\text{Vol}(D)}=\frac{n}{2\mathcal{s}}\frac{\Gamma\left(\frac{n}{2}+\frac12\right)}{\Gamma\left(\frac{n}{2}+1\right)}\sqrt{\pi},$$
where $S^+_x$ is the unit hemisphere in $T_x\overline{\mathcal{A}}$ and $D$ is the unit ball in $T_x\mathcal{C}$.
We can now conclude that, for all $\xi\in {\mathcal{E}}$ except for a set of zero probability,
\begin{align*}
\langle T\rangle_{\mathcal{E}}\frac{a}{V}&=\lim_{h\rightarrow 0}\lim_{m\rightarrow\infty}\left(\frac{T^m(\xi)}{m}\right)\left(\frac{1}{h T^m(\xi)}\int_0^{T^m(\xi)}\mathbbm{1}_{U_h}(\Phi_t(\xi))\, dt\right)\\
&=\lim_{h\rightarrow 0}\lim_{m\rightarrow\infty}\frac1{hm} \int_0^{T^m(\xi)} \mathbbm{1}_{U_h}(\Phi_t(\xi))\, dt\\
&=\lim_{h\rightarrow 0}\lim_{m\rightarrow\infty}\frac1{hm} \sum_{i=0}^{m-1}2\eta(R^i(\xi))\\
&=\lim_{h\rightarrow 0}\frac2h \langle \eta\rangle_{\mathcal{E}}\\
&=\frac{n}{\mathcal{s}}\frac{\Gamma\left(\frac{n}{2}+\frac12\right)}{\Gamma\left(\frac{n}{2}+1\right)}\sqrt{\pi}. 
\end{align*}
where we have used the notation $\Phi_t(\xi)$ for the billiard flow. This gives the second claimed identity. The third follows from the other three. The final claim is an immediate consequence of Kac's lemma,
noting that the ratio of the measures of $\mathcal{E}$ over that of $\mathcal{E}_0$ is the ratio of the volumes $D$ over that of $r_0D$, which is $1/r_0^{n-1}$.
\end{proof}

In the special case $n=2$, we have $$\langle T\rangle_{\mathcal{E}}=\frac{\pi}{\mathcal{s}}\frac{\text{area}(\mathcal{A})}{\text{length}(\mathcal{C})}.$$
 
\begin{corollary}\label{coro}
Suppose the value of the potential function in $\mathcal{A}_i$ is $C_i$ with $C_0>C_1$ and that the standard billiard system in $\overline{\mathcal{A}}_i$
is ergodic for $i=0$ and $1$.  Let $\langle T_i\rangle$ and $\langle N_i\rangle$ denote the mean time and number of collisions of the lensed billiard system during a sojourn 
in $\mathcal{A}_i$ before the next switch to the other region.
Then
$$\frac{\langle T_0\rangle}{\langle T_1\rangle}=\frac{V_0}{V_1} r_0^{n-2}, \ \ \frac{\langle N_0\rangle}{\langle N_1\rangle}=\frac{A_0}{A_1}r_0^{n-1}, $$
where $V_i$ and $A_i$ are the volume  of $\mathcal{A}_i$ and the volume of the boundary of $\mathcal{A}_i$.
\end{corollary}
\begin{proof}
A trajectory will transition from $\mathcal{A}_1$ to $\mathcal{A}_0$
whenever $(x,\theta)$ lies in the subset 
${\mathcal{E}}_0$ defined in Theorem \ref{averages}.
The particle speed in  $\mathcal{A}_i$ is $\mathcal{s}_i=\sqrt{2(E-C_i)/m}$, so
$\mathcal{s}_0/\mathcal{s}_1=r_0$.
The mean number of returns to $\mathcal{E}$ before the first return to $\mathcal{E}_0$ during a sojourn in $\mathcal{A}_1$ is $1/r_0^{n-1}$. Therefore
the claimed relations are a consequence of Theorem \ref{averages}.  
\end{proof}

 \subsection{Lensed billiard map differential}\label{Ap_D}
We give here the  proof of Theorem \ref{billiard differential}.
Let $t=t(x_1,\theta_1 )$ be such that $Q=q+tv$. Then
\begin{equation}\label{main_eq}\gamma_2(x_2) = \gamma_1(x_1) + t(x_1, \theta_1) e^{\theta_1 J}v_1. \end{equation}
Note that $\gamma_1'(x_1)=\mathbf{t}_1(\gamma_1(x_1))$ and  $\gamma_2'(x_2)=\mathbf{t}_2(\gamma_2(x_2)).$
Differentiating Equation (\ref{main_eq}) in $x_1$,
$$ \frac{\partial x_2}{\partial x_1}(0,0)\mathbf{t}_2(\mathcal{O}_2) = \mathbf{t}_1(\mathcal{O}_1) + \frac{\partial t}{\partial x_1}(0,0)v_1.$$
This implies
$$ \frac{\partial t}{\partial x_1}(0,0)=-{\mathbf{t}_1(\mathcal{O}_1)\cdot \mathbf{n}_2(\mathcal{O}_2)}/{v_1\cdot \mathbf{n}_2(\mathcal{O}_2)}$$
and 
\begin{align*}
 \frac{\partial x_2}{\partial x_1}(0,0) &=\mathbf{t}_1(\mathcal{O}_1)\cdot \mathbf{t}_2(\mathcal{O}_2)-\mathbf{t}_1(\mathcal{O}_1)\cdot\mathbf{n}_2(\mathcal{O}_2)\frac{v_1\cdot \mathbf{t}_2(\mathcal{O}_2)}{v_1\cdot \mathbf{n}_2(\mathcal{O}_2)}\\
 &=\frac{\left[\left(\mathbf{t}_2(\mathcal{O}_2)\wedge \mathbf{n}_2(\mathcal{O}_2)\right)\mathbf{t}_1(\mathcal{O}_1)\right]\cdot \nu_1}{\nu_1\cdot\mathbf{n}_2(\mathcal{O}_2)}\\
 &=\frac{\nu_1 \cdot \mathbf{n}_1(\mathcal{O}_1)}{\nu_1\cdot\mathbf{n}_2(\mathcal{O}_2)}.
 \end{align*}
 We have used the operation $(a\wedge b)c= (a\cdot c) b-(b\cdot c) a$ for vectors $a, b, c\in \mathbb{R}^n$. If $n=2$ and $(a, b)$ is a positive orthonormal basis of $\mathbb{R}^2$ then $J=a\wedge b$ is rotation counterclockwise by $\pi/2$. 
 Taking now the derivative in $\theta_1$ of both sides of Equation (\ref{main_eq}),
 $$\frac{\partial x_2}{\partial \theta_1}(0,0)\mathbf{t}_2(\mathcal{O}_2)= \frac{\partial t}{\partial \theta_1}(0,0)v_1 + t(0,0) Jv_1.$$
 Noting that $t(0,0)=\ell/|v_1|$, we obtain
 $$ \frac{\partial t}{\partial \theta_1}(0,0)=-\frac{\ell}{|v_1|} \frac{(J v_1)\cdot  \mathbf{n}_2(\mathcal{O_2})}{v_1\cdot\mathbf{n}_2(\mathcal{O_2})}    =
 -\frac{\ell}{|v_1|}
 \frac{\nu_1\cdot\mathbf{t}_2(\mathcal{O}_2)}{\nu_1\cdot \mathbf{n}_2(\mathcal{O}_2)}$$
 and 
 \begin{align*}
 \frac{\partial x_2}{\partial \theta_1}(0,0)&= -\frac{\ell}{|v_1|}
 \frac{\nu_1\cdot\mathbf{t}_2(\mathcal{O}_2)}{\nu_1\cdot \mathbf{n}_2(\mathcal{O}_2)} v_1\cdot \mathbf{t}_2(\mathcal{O}_2)
 +\frac{\ell}{|v_1|}(Jv_1)\cdot \mathbf{t}_2(\mathcal{O}_2)\\
 &=-\ell \frac{\nu_1\cdot \mathbf{t}_2(\mathcal{O}_2)}{\nu_1\cdot\mathbf{n}_2(\mathcal{O}_2)}\nu_1\cdot \mathbf{t}_2(\mathcal{O}_2)-\ell \nu_1\cdot \mathbf{n}_2(\mathcal{O}_2)\\
 &=-\frac{\ell}{\nu_1\cdot \mathbf{n}_2(\mathcal{O}_2)}\left[(\nu_1\cdot \mathbf{t}_2(\mathcal{O}_2))^2+(\nu_1\cdot \mathbf{n}_2(\mathcal{O}_2))^2\right]\\
 &=-\frac{\ell}{\nu_1\cdot \mathbf{n}_2(\mathcal{O}_2)}.
 \end{align*}
  We have obtained so far the first row of the differential of $\mathcal{T}$ for both reflection and refraction. For the second row,   the two cases must be treated separately. Let us first consider reflection. Then
  $V=v-2v\cdot \mathbf{n}_2(Q) \mathbf{n}_2(Q)$ which,  in the  $x_i, \theta_i$ coordinates, is
  \begin{equation}\label{main_angle}
e^{\theta_2 J}v_2 = e^{\theta_1 J}v_1 - 2\left(e^{\theta_1J}v_1\right)\cdot \mathbf{n}_2(Q(x_1,\theta_1)) \mathbf{n}_2(Q(x_1,\theta_1)), 
  \end{equation}
  where, by Equation (\ref{main_eq}), $\gamma_2(x_2)=Q(x_1,\theta_1)=\gamma_1(x_1)+t(x_1,\theta_1) e^{\theta_1J}v_1$.
  Differentiating Equation (\ref{main_angle}) in $x_1$ at $x_1=0, \theta_1=0$, yields
 \begin{equation}\label{eq3} \frac{\partial \theta_2}{\partial x_1}(0,0) Jv_2 = -2v_1\cdot \left(\frac{D\mathbf{n}_2}{\partial x_1}(\mathcal{O}_2)\right)\mathbf{n}_2(\mathcal{O}_2)- 2v_1\cdot \mathbf{n}_2(\mathcal{O_2}) \frac{D\mathbf{n}_2}{\partial x_1}(\mathcal{O}_2). \end{equation}
  Now $$\frac{D\mathbf{n}_2}{\partial x_1}(\mathcal{O}_2)=\frac{\partial x_2}{\partial x_1}(0,0)(D_{\mathbf{t}_2}\mathbf{n}_2)(\mathcal{O}_2)= 
  -\frac{\partial x_2}{\partial x_1}(0,0)\kappa(\mathcal{O}_2)\mathbf{t}_2(\mathcal{O}_2).$$
  Taking the dot product of Equation (\ref{eq3}) with $\mathbf{t}_2(\mathcal{O}_2)$, solving for $\frac{\partial \theta_2}{\partial x_1}(0,0)$, and substituting  the already obtained
  value of $\frac{\partial x_2}{\partial x_1}(0,0)$,  gives
  $$\frac{\partial \theta_2}{\partial x_1}(0,0)= -2 \kappa(\mathcal{O}_2)\frac{v_1\cdot\mathbf{n}_2(\mathcal{O}_2)}{v_2\cdot\mathbf{n}_2(\mathcal{O}_2)}\frac{\nu_1\cdot\mathbf{n}_1(\mathcal{O}_1)}{\nu_1\cdot \mathbf{n}_2(\mathcal{O}_2)}=-2 \kappa(\mathcal{O}_2)\frac{\nu_1\cdot\mathbf{n}_1(\mathcal{O}_1)}{\nu_2\cdot\mathbf{n}_2(\mathcal{O}_2)}.$$
  
  Next, we differentiate Equation (\ref{main_angle}) in $\theta_1$ at $x_1=0, \theta_1=0$:
  \begin{align}\label{eq4}
  \begin{split}
  \frac{\partial\theta_2}{\partial \theta_1}(0,0)Jv_2&= Jv_1 -2(Jv_1)\cdot \mathbf{n}_2(\mathcal{O}_2)  \mathbf{n}_2(\mathcal{O}_2) \\
   &\ \ \ \ \ \ \ \ \ \ \ \ \ 
  -2 v_1\cdot \frac{D\mathbf{n}_2}{\partial \theta_1}(\mathcal{O}_2)\mathbf{n}_2(\mathcal{O}_2)- 2v_1\cdot\mathbf{n}_2(\mathcal{O}_2) \frac{D\mathbf{n}_2}{\partial \theta_1}(\mathcal{O}_2).
  \end{split}
  \end{align}
  Note that
  $$\frac{D\mathbf{n}_2}{\partial \theta_1}(\mathcal{O}_2)= \frac{\partial x_2}{\partial \theta_1}(0,0)\left(D_{\mathbf{t}_2}\mathbf{n}_2\right)(\mathcal{O}_2)=
  -\frac{\partial x_2}{\partial \theta_1}(0,0)\kappa(\mathcal{O}_2)\mathbf{t}_2(\mathcal{O}_2).$$
  Taking the inner product of Equation \ref{eq4} with $\mathbf{t}_2(\mathcal{O}_2)$, substituting the already obtained $\frac{\partial x_2}{\partial \theta_1}(0,0)$,
  and solving for $ \frac{\partial\theta_2}{\partial \theta_1}(0,0)$, results in 
  $$\frac{\partial \theta_2}{\partial\theta_1}(0,0)=\frac{\nu_1\cdot \mathbf{n}_2(\mathcal{O}_2)}{\nu_2\cdot\mathbf{n}_2(\mathcal{O}_2)} + \frac{2\kappa(\mathcal{O}_2)\ell}{\nu_2\cdot \mathbf{n}_2(\mathcal{O}_2)}.$$
  This gives the differential when $\mathcal{T}$ produces a reflection. We now turn to the case of refraction, for which
  \begin{equation}\label{eq5}
  V=v\cdot\mathbf{t}_2(Q)\mathbf{t}_2(Q)+\left[\left(v\cdot \mathbf{n}_2(Q)\right)^2-\frac{2(C_2-C_1)}{m}\right]^{\frac12}\mathbf{n}_2(Q).
  \end{equation}
Using Equations (\ref{qQvV}) and differentiating Equation (\ref{eq5}) in $x_1$,
\begin{align}\label{long_eq}
\begin{split}
\frac{\partial \theta_2}{\partial x_1}(0,0) Jv_2 &= v_1\cdot \frac{D\mathbf{t}_2}{\partial x_1}(\mathcal{O}_2)\mathbf{t}_2(\mathcal{O}_2) + v_1\cdot \mathbf{t}_2(\mathcal{O}_2) \frac{D\mathbf{t}_2}{\partial x_1}(\mathcal{O}_2) \\
& \ \ \ \ \ \ \ \ \ +v_1\cdot\mathbf{n}_2(\mathcal{O}_2) \frac{v_1\cdot \frac{D\mathbf{n}_2}{\partial x_1}(\mathcal{O}_2)}{\left[({v}_1\cdot\mathbf{n}_2(\mathcal{O}_2)^2 -\frac{2(C_2-C_1)}{m})\right]^\frac{1}{2}}\mathbf{n}_2(\mathcal{O}_2)\\
& \ \ \ \ \ \ \ \ \  \ \ \ \ \ \ \ \ \  +\left[({v}_1\cdot\mathbf{n}_2(\mathcal{O}_2)^2 -\frac{2(C_2-C_1)}{m})\right]^\frac{1}{2}\frac{D\mathbf{n}_2}{\partial x_1}(\mathcal{O}_2).
\end{split}
\end{align}
 Noting that 
 $$\frac{D\mathbf{t}_2}{\partial x_1}(\mathcal{O}_2)= \kappa(\mathcal{O}_2) \mathbf{n}_2(\mathcal{O}_2) \frac{\partial x_2}{\partial x_1}(0,0), \ \   
 \frac{D\mathbf{n}_2}{\partial x_1}(\mathcal{O}_2)= -\kappa(\mathcal{O}_2) \mathbf{t}_2(\mathcal{O}_2) \frac{\partial x_2}{\partial x_1}(0,0),$$ 
 inserting the already obtained value for $\frac{\partial x_2}{\partial x_1}(0,0)$,  taking the inner product of Equation (\ref{long_eq}) with $\mathbf{t}_2(\mathcal{O}_2)$, and isolating
 $\frac{\partial \theta_2}{\partial x_1}(0,0)$ yields, after algebraic simplification,
 $$\frac{\partial \theta_2}{\partial x_1}(0,0)= -\kappa(\mathcal{O}_2)\sqrt{\frac{E-C_1}{E-C_2}}\frac{\nu_1\cdot \mathbf{n}_1(\mathcal{O}_1)}{\nu_2\cdot \mathbf{n}_2(\mathcal{O}_2)}
\alpha,$$
where 
$$ \alpha := 1 - \left[1 - \frac{C_2-C_1}{E-C_1}\frac1{(\nu_1\cdot \mathbf{n}_2(\mathcal{O}_2))^2}\right]^\frac{1}{2}.$$
  Finally, differentiating Equation (\ref{eq5}) in $\theta_1$ at $x_1=0, \theta_1=0$, taking the inner product with $\mathbf{t}_2(\mathcal{O}_2)$, using the already obtained partial derivative of $x_2$ with respect to $\theta_1$, replacing the derivatives of $\mathbf{t}_2$ and  $\mathbf{n}_2$ by expressions involving $\kappa$, and finally isolating $\frac{\partial \theta_2}{\partial \theta_1}(0,0)$ yields
 $$\frac{\partial \theta_2}{\partial \theta_1}(0,0) = \sqrt{\frac{E-C_1}{E-C_2}}\frac{\nu_1\cdot \mathbf{n}_2(\mathcal{O}_2)}{\nu_2\cdot \mathbf{n}_2(\mathcal{O}_2)}\left(1+\frac{\kappa(\mathcal{O}_2) \ell}{\nu_1\cdot\mathbf{n}_2(\mathcal{O}_2)}\alpha\right). $$
This is the last term of the differential of $\mathcal{T}$ that was  left to compute.


\begin{thebibliography}{Abcdef}

 \bibitem{Bal} 
\newblock  P.R. Baldwin,  
\newblock \emph{Soft Billiard Systems},
\newblock  Physica 29D (1988) 321-342.

\bibitem{MR2020220}
P.~B\'{a}lint and I.~P. T\'{o}th.
\newblock Correlation decay in certain soft billiards.
\newblock {\em Comm. Math. Phys.}, 243(1):55--91, 2003.

\bibitem{MR2191384}
P.~B\'{a}lint and I.~P. T\'{o}th.
\newblock Hyperbolicity in multi-dimensional {H}amiltonian systems with
  applications to soft billiards.
\newblock {\em Discrete Contin. Dyn. Syst.}, 15(1):37--59, 2006.
 
 \bibitem{bynk} 
\newblock  V.G. Baryakhtar, V.V. Yanovsky, S.V. Naydenov, and A.V. Kurilo,  
\newblock \emph{Chaos in composite billiards},
\newblock  Journal of Experimental and Theoretical Physics, 103(2): 292-302, 2006.

\bibitem{buni} 
\newblock  L.A. Bunimovich,  
\newblock \emph{On the ergodic properties of nowhere dispersing billiards},
\newblock  Commun. math. Phys. 65, 295-312 (1979)

\bibitem{BuniYao}
\newblock L.A. Bunimovich, Y. Su,
\newblock \emph{Back to Boundaries in Billiards},
\newblock \url{https://arxiv.org/abs/2203.00785}

\bibitem{dC}
\newblock M.P. do Carmo,
\newblock \emph{Riemannian Geometry}
\newblock{Birkh\"auser}, 1993.

 \bibitem{chernov} 
\newblock  N. Chernov, R. Markarian,  
\newblock \emph{Chaotic Billiards},
\newblock  Mathematical Surveys and Monographs, V. 127, AMS 2006.


\bibitem{CV} 
\newblock  Y. Colin de Verdi\'ere, 
\newblock \emph{The semi-classical ergodic theorem for discontinuous metrics},
\newblock   Volume 31 (2012-2014), p. 71-89


\bibitem{cook} 
\newblock S. Cook, R. Feres,
\newblock \emph{Random billiards with wall temperature and associated Markov chains},
\newblock Nonlinearity 25 (2012) 2503-2541.

\bibitem{MR1087385}
V.~Donnay and C.~Liverani.
\newblock Potentials on the two-torus for which the {H}amiltonian flow is
  ergodic.
\newblock {\em Comm. Math. Phys.}, 135(2):267--302, 1991.

\bibitem{MR1722985}
V.~J. Donnay.
\newblock Non-ergodicity of two particles interacting via a smooth potential.
\newblock {\em J. Statist. Phys.}, 96(5-6):1021--1048, 1999.

\bibitem{GG} 
\newblock R. Giamb\`o and F. Giannoni,
\newblock \emph{Minimal geodesics on manifolds with discontinuous metrics},
\newblock J. London Math. Soc. (2) 67 (2003) 527-544.


\bibitem{jssc} 
\newblock D. Jakobson, Y. Safarov, A. Strohmaier, Y.  Colin de Verdi\`ere,
\newblock \emph{The semiclassical theory of discontinuous systems and ray-splitting billiards},
\newblock  Amer. J. Math.  137 (2015), 859-906.

\bibitem{MR0885572}
A.~Knauf.
\newblock Ergodic and topological properties of {C}oulombic periodic
  potentials.
\newblock {\em Comm. Math. Phys.}, 110(1):89--112, 1987.

\bibitem{MR1012710}
A.~Knauf.
\newblock On soft billiard systems.
\newblock {\em Phys. D}, 36(3):259--262, 1989.

\bibitem{MR0433510}
I.~Kubo.
\newblock Perturbed billiard systems. {I}. {T}he ergodicity of the motion of a
  particle in a compound central field.
\newblock {\em Nagoya Math. J.}, 61:1--57, 1976.

\bibitem{MR1162355}
R.~Markarian.
\newblock Ergodic properties of plane billiards with symmetric potentials.
\newblock {\em Comm. Math. Phys.}, 145(3):435--446, 1992.



\bibitem{PeneSau}
\newblock F. P\`ene, B. Saussol,
\newblock \emph{Spatio-temporal Poisson processes for visits to small sets},
\newblock Isr. J. Math. 240 (2020), 625-665.

\bibitem{MR2304468}
A.~Rapoport, V.~Rom-Kedar, and D.~Turaev.
\newblock Approximating multi-dimensional {H}amiltonian flows by billiards.
\newblock {\em Comm. Math. Phys.}, 272(3):567--600, 2007.

\bibitem{MR2383597}
A.~Rapoport, V.~Rom-Kedar, and D.~Turaev.
\newblock Stability in high dimensional steep repelling potentials.
\newblock {\em Comm. Math. Phys.}, 279(2):497--534, 2008.

\bibitem{MR1692874}
V.~Rom-Kedar and D.~Turaev.
\newblock Big islands in dispersing billiard-like potentials.
\newblock {\em Phys. D}, 130(3-4):187--210, 1999.

\bibitem{MR1997267}
D.~Turaev and V.~Rom-Kedar.
\newblock Soft billiards with corners.
\newblock {\em J. Statist. Phys.}, 112(3-4):765--813, 2003.


\end{thebibliography}
 \end{document}